\documentclass[11pt,a4paper]{article}

\usepackage[english]{babel} 
\usepackage[final]{microtype} 

\usepackage{amssymb,amsmath,mathtools} 
\usepackage{graphicx} 
\usepackage{bbm} 
\usepackage[titletoc,title]{appendix}
\usepackage{cite}
\usepackage[font=small,labelfont=bf]{caption}
\usepackage{relsize}

\usepackage[hmargin=0.12\paperwidth,vmargin=0.2\paperwidth,bindingoffset=0cm]{geometry}

\usepackage{titlesec}

\titleformat*{\section}{\large\bfseries}
\titleformat*{\subsection}{\bfseries}

\usepackage{amsthm}

\theoremstyle{plain}
  \newtheorem{theorem}{Theorem}[section]
  \newtheorem{prop}[theorem]{Proposition}
  \newtheorem{lemma}[theorem]{Lemma}
  
\theoremstyle{definition}

\theoremstyle{remark}

\newcommand{\p}{\partial} 

\newcommand{\nn}{\nonumber} 

\usepackage{tensor}
\def\bea{\begin{eqnarray}}
\def\eea{\end{eqnarray}}
\newcommand{\beq}{\begin{equation}}
\newcommand{\eeq}{\end{equation}}

\DeclareMathOperator{\Tr}{Tr}

\newcommand{\al}{\alpha}
\newcommand{\be}{\beta}

\renewcommand{\phi}{\varphi} 
\newcommand{\ep}{\epsilon}

\newcommand{\la}{\lambda}

\newcommand{\MeijerG}[8][\Big]{G^{{ #2 },{ #3 }}_{{ #4 },{ #5 }} #1
( \begin{matrix} #6 \\ #7 \end{matrix}\, #1\vert\, #8 #1)}
\newcommand{\hypergeometric}[6][\bigg]{\,{}_{#2} F_{#3} #1
( \begin{matrix} #4 \\ #5 \end{matrix}\, #1\vert\, #6 #1)}
\newcommand{\fm}[2]{\,{}_{#1} F_{#2}}

\newcommand{\ff}{\tensor*[_{1}]{{F}}{_1}}

\newcounter{app}
\newcounter{sapp}[app]

\newcommand{\ds}{\displaystyle}

\def\nsection#1{\setcounter{equation}{0}\section{#1}}

\title{\bfseries Integrable structure of products of finite complex Ginibre random matrices}
\author{\Large  Vladimir
  V.~Mangazeev$^{1}$ and Peter J. Forrester$^{2}$}

\begin{document}

\date{\today}


\maketitle

\begin{center}
$^1$\large
Department of Theoretical Physics,
         Research School of Physics and Engineering,\\
    Australian National University, Canberra, ACT 0200, Australia
\end{center}
\begin{center}{$^2$\large
School of Mathematics and Statistics, \\
ARC Centre of Excellence for Mathematical
 and Statistical Frontiers, \\
 The University of Melbourne,
Victoria 3010, Australia}
\end{center}

\begin{abstract}
We consider the squared singular values of the product of $M$ standard
complex Gaussian matrices. Since the
squared singular values form a determinantal point process with a particular Meijer
G-function kernel, the gap probabilities are given by a Fredholm determinant based on this
kernel.
 It was shown by Strahov \cite{St14}
that a hard edge
scaling limit of the gap probabilities is described by Hamiltonian differential equations which
can be formulated as an isomonodromic deformation system similar to
the theory of the Kyoto school. We generalize this result
to the case of finite matrices by first finding a representation of the finite
kernel in integrable form. As a result we obtain
the Hamiltonian structure for a finite size matrices and formulate it
in terms of a $(M+1) \times (M+1)$ matrix Schlesinger system. The case $M=1$ reproduces
the Tracy and Widom theory which results in the Painlev\'e V equation for
the $(0,s)$ gap probability.  Some integrals of motion for $M = 2$ are identified, and a
coupled system of differential equations in two unknowns is presented which uniquely determines the
corresponding $(0,s)$ gap probability.

\end{abstract}

\newpage

\nsection{Introduction}
Consider a point process on the line. The process is said to be determinantal
if the $k$-point correlation functions $\rho_{(k)}$ have the form
\begin{equation}\label{a0}
\rho_{(k)}(x_1,\dots, x_k) = \det [ K(x_i, x_j) ]_{i,j=1}^k,
\end{equation}
for $K(x,y)$ --- the so-called correlation kernel ---independent of $k$.
The eigenvalues of many ensembles of complex Hermitian matrices,
and their various scaling limits are well known examples of determinantal point processes,
as are the positions of nonintersecting random walkers on the line; see e.g.~the monographs
\cite[Ch.~5]{Forrester2010} and \cite{Ka15}.

For a one-dimensional point process, let $E(k;J)$ denote the probability that
there are exactly $k$ eigenvalues in the interval $J$.
With a slight abuse of notation, introduce the generating function
\begin{equation}\label{a1}
E(\lambda;J) = \sum_{k=0}^\infty (1 - \lambda)^k E(k;J).
\end{equation}
A characterising feature of the determinant case is that (\ref{a1}) can be expressed
as a Fredholm determinant
\begin{equation}\label{a2}
E(\lambda;J) = \det ( \mathbb I - \lambda \mathbb K_J).
\end{equation}
Here $\mathbb K_J$ denotes the integral operator on $J$ with kernel $K(x,y)$,
as appears in (\ref{a0}).

Suppose furthermore that the correlation kernel has the additional structure
\begin{equation}\label{a3}
K(x,y) = {\sum_{i=1}^r \frac{f_i(x) g_i(y)}{x - y}},
\end{equation}
where $\sum_{i=1}^r f_i(x) g_i(x) = 0$. Kernels of the form (\ref{a3}) are
termed integrable in \cite{IIKS90}. They have the general property that the
corresponding resolvent kernel
is also an integrable kernel.
The simplest case of (\ref{a3}) occurs when $r=2$ and $f_2 = f_1$, $g_2 = - g_1$, giving
\begin{equation}\label{a4}
K(x,y) = \frac{f(x) g(y) - g(x) f(y)}{x - y}.
\end{equation}
This is well known in random matrix theory. It results from  unitary invariant
ensembles, as a consequence of
the Christoffel-Darboux summation formula (see e.g.~\cite[Prop.~5.1.3]{Forrester2010}).
For example, with
\begin{equation}\label{a5}
f(x) = \frac{1}{\pi} \sin \pi x, \quad g(x) = \frac{1}{\pi} \cos \pi x,
\end{equation}
(\ref{a4}) gives the  sine kernel
\begin{equation}\label{a6}
K(x,y) = \frac{\sin \pi (x - y)}{\pi (x - y) },
\end{equation}
which is the correlation kernel for complex Hermitian random matrices with
bulk scaling (see e.g.~\cite[Ch.~5]{Forrester2010}).

Note that (\ref{a5}) satisfies the first order matrix linear differential equation
\begin{equation}\label{a7}
{1 \over \pi} {d \over dx} \begin{bmatrix} f(x) \\ g(x) \end{bmatrix} =
 \begin{bmatrix} 0 & 1 \\ -1 & 0 \end{bmatrix}  \begin{bmatrix} f(x) \\ g(x)
 \end{bmatrix} .
\end{equation}
Tracy and Widom \cite{TW94c} shows that for kernels (\ref{a4}) with $(f,g)$
satisfying the first order matrix linear differential
equation (\ref{a7}) corresponding to classical orthogonal polynomials or their
scaling limits, quantities associated with Fredholm
determinant (\ref{a2}) satisfy an integrable (Hamiltonian) system of non-linear
differential equations. For certain intervals $J$
depending on a single parameter, this system could be integrated to yield
a characterisation of the logarithmic derivative
of (\ref{a2}) as the solution of a Painlev\'e equation in sigma form
(see e.g.~\cite[\S 8.1]{Forrester2010}). This work generalised, and in
fact was inspired by, the work of the Kyoto school \cite{JMMS80} in the case
of the sine kernel (\ref{a6}), in which results of
this type were first derived. See \cite[Ch.~9]{Forrester2010} for a text book treatment.
At this point it is worth mentioning other remarkable apperances of Painlev\'e
equations in the theory of integrable systems \cite{McCoy76,Zam94,BM06}.

The first study in random matrix theory to give rise to a kernel of the form
(\ref{a3}) with $r=3$ was that of the so-called
Pearcey kernel \cite{BH98,TW06,BK06}. It comes about as the critical scaling
of the matrix sum $t H + H_0$, where $H$ is a member
of the GUE (complex Hermitian random matrices) and $H_0$ is a fixed matrix with
half its eigenvalues at $+1$ and the
other half at $-1$; $t$ is a parameter.
The $f_i,g_i$ in (\ref{a3}) satisfy third order linear differential equations, and
Br\'ezin and Hikami \cite{BH98}
showed that the method  of \cite{TW94c} could be adapted to this setting,
obtaining a characterisation of the gap probability for $J$ a symmetrical interval
about the origin in terms of a pair of coupled
nonlinear equations. For the parameter dependent extension (the Pearcey process),
the kernel is again of the form (\ref{a3}) with $r=3$ and the
 $f_i,g_i$ satisfying third order linear differential equations. PDEs for the corresponding
 gap probabilities have been derived in \cite{AvM07}, and their numerical evaluation using
 the method of Bornemann \cite{Bo09} has been studied in \cite{BC13}.

More recently, the hard edge scaling (see Section \ref{S2.3}) of the squared
singular values of $M$ standard complex rectangular Gaussian random matrices
has been shown to be of the form (\ref{a3}) with $r=M+1$ \cite{KZ14}.
From this, Strahov \cite{St14}
generalised the approach of Tracy and Widom to derive a system of nonlinear partial
differential equations associated
with (\ref{a2}) in the case that $J$ is given by a disjoint union of positive
intervals $(a_{2j-1}, a_{2j})$, $j=1,\dots,m$,
\beq\label{J}
J =  \cup_{j=1}^m(a_{2j-1}, a_{2j})
\eeq
 He also
found the Hamiltonian of the associated dynamical system and derived its isomonodromic
representation. For a single
interval $J=(0,s)$, Witte and Forrester \cite{WF17} showed how these coupled equations
could be integrated in the case $M=1$
to reclaim results obtained originally in \cite{TW94b} for the Bessel kernel
(this reduction was also investigated in
\cite{St14}). The same task was carried out in the case $M=2$, leading to
the characterisation of the Fredholm determinant
in terms of the solution of a certain fourth order nonlinear ordinary differential
equation. The latter is lengthy; for a special
choice of parameters a much simpler third order equation was found upon the basis
of series expansions, but a proof has yet to be found.
Notwithstanding its complex nature, as an application the 4th order
equation was used to deduce the leading large $s$ form of the gap
probability. In a recent development, Claeys, Givotti and Stivigny \cite{CGS16}
used a Riemann Hilbert analysis to extend this result to the first three orders,
and also to general $M$.

The Bessel kernel as is relevant to the case $M=1$ results as a hard edge scaling limit
of the Laguerre kernel \cite{Fo93c}.
Tracy and Widom \cite{TW94c} applied their theory directly to the Laguerre kernel,
and integrated the resulting system of
coupled nonlinear equations in the case $J=(0,s)$ to obtain a characterisation of
(\ref{a2}) in terms of the solution of a
$\sigma$-Painlev\'e V equation; see also \cite{AvM95,FW01a,FO10,BC09}.
Taking the hard edge scaling limit of the latter
directly gives the $\sigma$-Painlev\'e III equation
characterising the Bessel kernel.

This motivates us to embark on an analogous study of the finite matrix sizes kernel
for the squared singular values of the product of $M$
rectangular complex Hermitian matrices. Specifically, in this work we show that
the corresponding kernel can be written in
integrable form (\ref{a3}), and that the analogue of Strahov's equations can be derived.
These equations can be written in
Hamiltonian form, and as the isomonodromic deformation of a linear system.
For $M=1$ it is shown that they are equivalent to the system of equations for the
gap probabilities associated with the Laguerre kernel, as isolated by Tracy and Widom \cite{TW94c}.
For $M=2$ several integrals of motion are deduced. Moreover, a coupled differential system in two unknowns
is presented which uniquely determines the gap probability for no eigenvalues in $(0,s)$.

\nsection{Singular values of products of complex Ginibre random matrices}
\subsection{The kernel}
Complex Ginibre matrices are random matrices with independent standard complex Gaussian
entries. Let $X_1,\dots, X_M$, $M \ge 1$ be a sequence of such matrices with $X_m$
of size $N_m \times N_{m-1}$ ($1 \le m \le M$), and define the product
\beq  \label{sing1}
Y_M = X_M X_{M-1} \cdots X_1.
\eeq
That the squared singular values of $Y_M$, or equivalently the eigenvalues of
$Y^{\dagger}_M Y_M$, $\mbox{Spec}(Y^\dagger_M Y_M)=(x_1,\ldots,x_n)$ form a determinantal
point process on $\mathbb{R}_{>0}$ was
 first established by Akemann, Ipsen and Kieburg \cite{AIK13}, and further insights
 were given by Kuijlaars and Zhang
 \cite{KZ14}. The work \cite{AIK13}
 extended that of
Akemann, Kieburg   and Wei \cite{AKW13} in the case that each $X_m$ is square.
A review of these recent developments is given in \cite{AI15}.
Here we record the explicit form of the
correlation kernel, which is given in terms of Meijer G-functions
(see the Appendix for the definition).

\begin{theorem}
Introduce the parameters
\beq \nu_m=N_m-N_0,\quad \nu_m\geq0 \quad m=0,\ldots,M, \quad n=N_0 \label{sing2}
\eeq
In terms of the Meijer G-function \cite{BS13,Luke69} define
\beq
Q_n(x)=\frac{1}{\ds n! \prod_{j=1}^M\Gamma(\nu_j+n+1)}
{G}^{M+1,0}_{1,M+1}
\left.\left(\begin{array}{c}-n \\0,\nu_1,\ldots,\nu_M\end{array}\right|x\right) \label{ker3}
\eeq
(this is eq.~(3.7) of \cite{KZ14} and eq.~(47) of \cite{AIK13}) and
\beq \label{ker2}
P_n(x)=
-{\prod_{j=0}^M\Gamma(\nu_j+n+1)}\>\>
{G}^{1,0}_{1,M+1}
\left.\left(\begin{array}{c}n+1 \\-\nu_0,-\nu_1,\ldots,-\nu_M\end{array}\right|x\right)
\eeq
(this is  eq.~(44) of \cite{AIK13}
and eq.~(3.11) of \cite{AIK13}).
The correlation kernel for the determinantal point process specifying the
statistical distribution of $\mbox{Spec}(Y^\dagger_M Y_M)$ is given by \cite{AIK13}
\beq
K^M_n(x,y)=\sum_{k=0}^{n-1}P_k(x)Q_k(y)  \label{sing11}
\eeq
or alternatively \cite[eq.~(5.4)]{KZ14}
\footnote{There is a misprint in (5.4) of \cite{KZ14} where
the product should start from $j=0$ instead of $j=1$.}
\beq
K_n^M(x,y)=-\prod_{j=0}^M(n+\nu_j)\int_{0}^1P_{n-1}(ux)Q_n(uy)du.\label{edge6}
\eeq
\end{theorem}

We remark that $P_n(x)$ is a polynomial of degree $n$,
and as revised in Appendix A,
it can alternatively be written as a generalised hypergeometric function.

The singular values of products of complex Ginibre matrices is one of
a number of random matrix ensembles which gives rise to a correlation
kernel of the form (2.5), with $P_n(x)$, $Q_n(y)$ given in terms of
Meijer G-functions. Others include the Cauchy two-matrix model
\cite{BGS14}, the closely related Bures ensemble of random density
matrices \cite{FK16}, the singular values of products of complex
Ginibre matrices and their inverses \cite{Fo14}, and the singular values
of products of truncated unitary matrices \cite{KKS16}.

\subsection{Properties of biorthogonal functions $P_n(x)$ and $Q_n(x).$}\label{S33}

Following [16, 29], for future use we make note of several properties of 
the biorthogonal functions $P_n(x)$ and $Q_n(x)$,
following essentially from their definition as Meijer G-functions.
 We start with

\begin{prop}\label{prop1}
Let $\delta_x = x {d \over dx}$. We have
\bea
&&\prod_{i=0}^M(\delta_x+\nu_i) P_n(x)=x(\delta_x-n)P_n(x), \label{ker8} \\
&&\prod_{i=0}^M(\delta_x-\nu_i) Q_n(x)=(-1)^M x(\delta_x+n+1)Q_n(x). \label{ker9}
\eea
\end{prop}
The proof  follows from the differential equation for Meijer G-functions \eqref{A4}.

\begin{prop}\label{prop1a}
Upon multiplication by $x$,
$P_n(x)$ and $Q_n(x)$ satisfy the recurrence relation
\begin{align}
&xP_{n}(x)=P_{n+1}(x)+\sum_{k=0}^M P_{n-k}(x)\>a_{k,n},\label{ker9a}\\
&xQ_{n}(x)=Q_{n-1}(x)+\sum_{k=0}^M Q_{n+k}(x)\>a_{k,n+k},\label{ker9b}
\end{align}
where
\beq
a_{k,n}=\prod_{j=0}^M(n-k+\nu_j+1)_k\>\>\sum_{j=0}^{k+1}\frac{(-1)^{j}}{j!(k+1-j)!}
\prod_{i=0}^M(n+1-j+\nu_i). \label{ker9c}
\eeq
\end{prop}
This proposition was proved in \cite[Section 4]{KZ14}.

Next we note
\begin{prop}\label{prop2}
Upon application of the operator $\delta_x$, $P_n(x)$ and $Q_n(x)$ satisfy the recurrence
\begin{align}
 \prod_{i=0}^M(n+\nu_i)\> P_{n-1}(x)=&(\delta_x-n)P_n(x), \label{ker10} \\
\prod_{i=0}^M(n+\nu_i+1)\> Q_{n+1}(x)=&(-\delta_x-n-1)Q_n(x). \label{ker11}
\end{align}
\end{prop}

\begin{proof}
From \cite[Eq.~(3.8)]{KZ14} we have
\beq
P_n(x)=\frac{1}{2\pi i}\>{\ds\prod_{j=0}^M\Gamma(n+\nu_j+1)}
\oint\limits_{\Sigma_n}
\frac{\Gamma(t-n)}{\ds\prod_{j=0}^M\Gamma(t+\nu_j+1)}x^t dt, \label{ker12}
\eeq
where $\Sigma_n$ is a closed contour encircling $0,\ldots,n$ in a positive direction.

Let us calculate the RHS of \eqref{ker10}. We have the identity
\beq
(\delta_x-n)[\Gamma(t-n)x^t]=\Gamma(t-(n-1))x^t. \label{ker13}
\eeq
Therefore, the pole of the integrand at $t=n$ disappears, the contour shrinks to
$\Sigma_{n-1}$ and we come to
to the integral representation for $P_{n-1}(x)$.
The extra factor in the LHS of \eqref{ker10} comes from the pre-factor in \eqref{ker12}.

Similarly for $Q_n(x)$ we have from (3.6) in \cite{KZ14}
\beq
Q_n(x)=\frac{1}{2\pi i\>\ds\prod_{j=0}^M\Gamma(n+\nu_j+1)}\>\>
{\int\limits_{-i\infty}^{i\infty} }\>\frac{\ds\prod_{j=0}^M\Gamma(t+\nu_j)}{\Gamma(t-n)}x^{-t} dt.
\label{ker14}
\eeq
Using the identity
\beq
(-\delta_x-n-1)\left[\frac{x^{-t}}{\Gamma(t-n)}\right]=\frac{x^{-t}}{\Gamma(t-(n+1))} \label{ker15}
\eeq
we immediately obtain \eqref{ker11}.

\end{proof}

A generalisation of  {\bf Proposition \ref{prop2}} is
\begin{prop}\label{prop3}
We have
\begin{align}
&\prod_{j=0}^M(n-m+\nu_j+1)_m \>P_{n-m}(x)=
(\delta_x-n)_m \>P_n(x),\label{ker18}\\
&\prod_{j=0}^M(n+\nu_j+1)_m\> Q_{n+m}(x)=
(-1)^m(\delta_x+n+1)_m\>Q_n(x).\label{ker19}
\end{align}

\end{prop}

\begin{proof}

Let us prove  \eqref{ker18} by induction in $m$.
For $m=1$ \eqref{ker18} coincides with \eqref{ker10}.
Consider the LHS of \eqref{ker18} with $m$ replaced with $m+1$.
Using \eqref{ker10} with $n$ replaced by $n-m$
 we obtain
\beq
\prod_{j=0}^M(n-m+\nu_j)_{m+1} \>P_{n-m-1}(x)=
\prod_{j=0}^M(n-m+\nu_j+1)_m \>(\delta_x+m-n)P_{n-m}(x).\label{ker22}
\eeq
Now applying  \eqref{ker18} to the RHS of \eqref{ker22} we get
\beq
(\delta_x+m-n)\prod_{j=0}^M(n-m+\nu_j+1)_m \>P_{n-m}(x)=
(\delta_x+m-n)(\delta_x-n)_m \>P_n(x)=(\delta_x-n)_{m+1} \>P_n(x), \label{ker22a}
\eeq
which completes the proof.
 The proof of \eqref{ker19} is similar.
 \end{proof}

 An identity involving multiplication by $x$ and the operator $\delta_x$ is also of interest.

\begin{prop}\label{prop4}
We have
\begin{align}
P_n(x)-xP_{n-1}(x)+&\sum_{k=0}^M\sum_{l=0}^{M-k}
e_{M-k-l}({\boldsymbol \nu})n^l\delta_x^k P_{n-1}(x)=0, \label{ker28} \\
Q_{n-1}(x)-xQ_{n}(x)+&\sum_{k=0}^M\sum_{l=0}^{M-k}
e_{M-k-l}({\boldsymbol \nu})n^l(-\delta_x)^k Q_{n}(x)=0, \label{ker29}
\end{align}
where $e_k({\boldsymbol \nu})$ is the $k$-th elementary symmetric function
of $M$ variables ${\boldsymbol \nu}=(\nu_1,\ldots,\nu_M)$.

\begin{proof}
Let us first write \eqref{ker28} in the form
\beq
P_n(x)= \Big ( x - {\mathcal A}_n(\delta_x)\Big )P_{n-1}(x), \label{ker30}
\eeq
where ${\mathcal A}_n(\delta_x)$ is the differential operator given by a double sum in \eqref{ker28}.
Using \eqref{ker10} we can rewrite \eqref{ker30} in the form
\beq
\left[{x(\delta_x-n)-\mathcal A}_n(\delta_x)(\delta_x-n)-\prod_{j=0}^M (n+\nu_j)\right]P_n(x)=0.
\label{ker31}
\eeq
Using \eqref{ker8} we obtain from \eqref{ker31}
\beq
\left[\prod_{j=0}^M(\delta_x+\nu_j)-\prod_{j=0}^M (n+\nu_j)-{\mathcal A}_n(\delta_x)(\delta_x-n)\right]
P_n(x)=0. \label{ker32}
\eeq
Let us show that the differential operator in the LHS is identically equal to $0$. We have
\begin{align}
&{\mathcal A}_n(\delta_x)(\delta_x-n)=
\sum_{k=0}^M\sum_{l=0}^{M-k}
e_{M-k-l}(\nu_1,\ldots,\nu_M)(n^l\delta_x^{k+1}-n^{l+1}\delta_x^k)=\nonumber\\
&=\sum_{k=0}^{M+1}e_{M+1-k}(0,\nu_1,\ldots,\nu_M)(\delta_x^k-n^k)=
\prod_{j=0}^M(\delta_x+\nu_j)-\prod_{j=0}^M (n+\nu_j). \label{ker33}
\end{align}
Here we used the fact that the double sum in \eqref{ker33} is telescopic and so only
boundary one-dimensional sums  survive. Thus, \eqref{ker32} is proved.
The proof of \eqref{ker29} follows from \eqref{ker9} and \eqref{ker11} in a similar manner.

\end{proof}

\end{prop}


\subsection{The hard edge  limit}\label{S2.3}

In the limit $N_0\to\infty$ the eigenvalues near the origin are spaced at distances of order $1/N_0$.
Changing scale $x_j \mapsto x_j/N_0$ and taking $N_0 \to \infty$ with $\nu_m$ fixed defines the hard edge
limit, and the corresponding hard edge scaled correlation function
\beq
\lim\limits_{N_0\to\infty}\frac{1}{N_0}K_{N_0}^M\left(\frac{x}{N_0},\frac{y}{N_0}\right)=K^M(x,y),
\label{edge1}
\eeq
is well defined. Kuijlaars and Zhang \cite{KZ14} used (\ref{edge6}) to deduce that
\beq
K^M(x,y)=\int_0^1 P(ux) Q(uy) du, \label{edge7}
\eeq
with $P(x)$ and $Q(y)$ defined by
\beq
P(x)=\MeijerG{1}{0}{0}{M+1}{-}{-\nu_0,-\nu_{1},\ldots,-\nu_M}{x},\quad
Q(y)=\MeijerG{M}{0}{0}{M+1}{-}{\nu_1,\ldots,\nu_M,\nu_0}{y}. \label{edge2a}
\eeq
Most significant for our present purposes is that
these authors were able to deduce from (\ref{edge7}) that $K^M$ can be written as
an integrable kernel.

\begin{theorem}\label{thedge1}
Let $P(x)$ and $Q(y)$ be given by (\ref{edge2a}). Let
$\mathcal{B}(.,.)$ be the  bilinear operator defined by
\beq
\mathcal{B}(f(x),g(y))=(-1)^{M+1}\sum_{j=0}^M(-1)^j\delta_x^jf(x)
\sum_{i=0}^{M-j}\alpha_{i+j}\delta_y^jg(y),
\quad \delta_x\equiv x\frac{d}{dx},
\label{edge3}
\eeq
 with the constants $\al_i$
given  by
\beq
\prod_{i=1}^M(x-\nu_i)=\sum_{i=0}^M\al_i x^i, \label{edge4}
\eeq
or equivalently in terms of an elementary symmetric function
\beq
\al_i=(-1)^i e_{M-i}(\nu_1,\ldots,\nu_M).\label{edge5}
\eeq
We have
\begin{equation}\label{2.2}
K^M(x,y) = {  \mathcal{B}(P(x), Q(y)) \over x - y}.
\end{equation}
\end{theorem}

\nsection{The integrable form of the kernel $K_n^M(x,y)$}

We would like to express the finite $n$ kernel (\ref{sing11}) in integrable form.
In light of the fact that the hard edge scaled kernel $K^M(x,y)$ was derived from
the integral representation (\ref{edge7}), it seems natural to start from the representation
(\ref{edge6}), and to use the differential equations for $P_k(x)$ and $Q_k(x)$ analogous to
what was done in \cite{KZ14}; see also \cite{Zh17} in the closely related case of
the hard edge scaled Muttalib--Borodin model \cite{Mu95,Bor99,FW15}. However,
the presence of the parameter $n$ makes it
unclear as to how to implement this strategy.

We proceed instead by algebraic means. Our central result is

\begin{theorem}\label{Th1}
The kernel $K_n^M(x,y)$ permits the integrable form
\beq
K_n^M(x,y)=\frac{{\mathcal{D}}(P_n(x),Q_n(y))}{x-y}, \label{repr}
\eeq
valid for any $n \ge 1$,
where the bilinear differential operator ${{\mathcal{D}}}$
does not depend on $n$ and has the form
\beq
{\mathcal{D}}(f(x),g(y))=
\sum_{j=0}^M \phi_j(x)\psi_j(y)  \label{ker5}
\eeq
with
\beq\phi_j(x)=(-1)^{j+1}\delta_x^j f(x), \quad
\psi_j(x)=-\delta_{j,0}\,x g(x)+\sum_{i=0}^{M-j}\alpha_{i+j}\delta_x^i g(x),
 \label{ker6}
\eeq
and where  the $\alpha_i$ are given by (\ref{edge5}).
\end{theorem}

The simplicity of this result for any finite $n$ is striking. Comparing the bilinear
differential operators $\mathcal{B}$ from \eqref{edge3} and $\mathcal{D}$ from \eqref{ker5}
we see that they are almost identical except for the overall factor $(-1)^{M}$ and the
extra term in $\psi_0(x)$ in \eqref{ker6}.

To prove the above theorem we need some preparatory lemmas.
\begin{lemma} \label{lem1}
For any $n\geq 0$ and $x,y\in\mathbb{C}$
\beq
\sum_{i=0}^n(-1)^i \frac{(x-i)^n}{i!(n-i)!(y+i)}=\frac{(x+y)^n}{(y)_{n+1}}\label{ker40a}
\eeq
and
\beq
\sum_{i=0}^n(-1)^i \frac{(x-i)^{n+1}}{i!(n-i)!(y+i)}=\frac{(x+y)^{n+1}}{(y)_{n+1}}-1.\label{ker40a1}
\eeq
\end{lemma}
\begin{proof}
Consider first (\ref{ker40a}), and regard both sides as a function of the complex variable $y$.
Both sides go to zero as $|y| \to \infty$ and have simple poles at
$y=0,-1,\dots,-n$ with the same residues, and hence are identical functions of $y$.
The same argument works for (\ref{ker40a1}).
\end{proof}

\begin{lemma}\label{lem2} For $x,y,z\in\mathbb{C}$ and $l\geq0$

\beq
\sum_{k=1}^l\sum_{m=0}^{k-1}\sum_{j=0}^{k+1}
(-1)^{j+m}\frac{(x)_{k-m}(y)_m}{j!(k+1-j)!}(z+m-j+1)^{l+1}=
\frac{z^{l+1}}{1-y}-\frac{(x+z)^{l+1}}{1-x-y}+\frac{x(1-y+z)^{l+1}}{(1-y)(1-x-y)}.
\label{ker40}
\eeq
\end{lemma}

\begin{proof}
Let us change the order of summations in $k$ and $m$ and introduce a new variable $s=k-m$.
Then we can rewrite the LHS of \eqref{ker40} as
\beq
\sum_{m=0}^{l-1}\sum_{s=1}^{l-m}\sum_{j=0}^{l+1}
(-1)^{j+m}\frac{(x)_{s}(y)_m}{j!(s+m+1-j)!}(z+m-j+1)^{l+1}. \label{ker41}
\eeq
We extended the summation in $j$ to $l+1$ since it is truncated by the factor $(s+m+1-j)!$
in the denominator.
If we interchange the summations in $s$ and $j$,
we need to split the sum in $j$ at the value $j=m+1$
and write
\beq
\sum_{s=0}^{l-m}\sum_{j=0}^{l+1}=\sum_{j=0}^{m+1}\sum_{s=1}^{l-m}+
\sum_{j=m+2}^{l+1}\>\sum_{s=j-m-1}^{l-m}, \label{ker43}
\eeq
where we take into account the truncating condition $s+m+1-j\geq 0$ in \eqref{ker41}.

Now the sum in $s$ can be calculated. The simple identity
\beq
\sum_{s=0}^k \frac{(x)_s}{s!}=\frac{(x+1)_k}{k!} \label{ker44}
\eeq
shows
\beq
\sum_{s=m+1}^k \frac{(x)_s}{s!}=\frac{(x+1)_k}{k!}-\frac{(x+1)_m}{m!},
\quad k\geq m.\label{ker45}
\eeq
Using \eqref{ker45} we can calculate the sums in $s$ in both terms in \eqref{ker43}.
This allows (\ref{ker41}) to be reduced to two double summations
\begin{align}
&\sum_{m=0}^{l-1}\sum_{j=0}^{l+1}(-1)^{j+m}\frac{(x)_{l-m+1}(y)_m}{j!(1+l-j)!(j-m+x-1)}
(1+z-j+m)^{l+1}\label{ker46a}\\
+&\sum_{m=0}^{l-1}\sum_{j=0}^{m+1}(-1)^{j+m+1}\frac{x\,(y)_m}{j!(1+m-j)!(j-m+x-1)}
(1+z-j+m)^{l+1}.\label{ker46}
\end{align}

Let us consider the first sum  \eqref{ker46a}. The sum in $j$ can be evaluated using
Lemma \ref{lem1}
\beq
\sum_{j=0}^{l+1}(-1)^{j}\frac{(1+z-j+m)^{l+1}}{j!(1+l-j)!(j-m+x-1)}=
\frac{(x+z)^{l+1}}{(x-m-1)_{l+2}}. \label{ker47}
\eeq
To calculate the sum over $m$ in \eqref{ker46a} we need
the formula
\beq
\sum_{m=0}^n\frac{(x)_m}{(y)_m}=\frac{y-1}{1+x-y}+\frac{(x)_{n+1}}{(1+x-y)(y)_n}, \label{ker48}
\eeq
which is easy to prove by induction.
Using \eqref{ker48} we finally get the answer for the sum \eqref{ker46a}
\beq
-\frac{(x+z)^{l+1}}{1-x-y}+(-1)^l \frac{  \>(x)_{-l}(y)_l(x+z)^{l+1}}{1-x-y}.\label{ker49}
\eeq
Now let us turn to the second sum \eqref{ker46}. Changing the summation variable $j=m+1-r$
shows
\beq
\sum_{m=0}^{l-1}\sum_{r=0}^{m+1}(-1)^{r}\frac{x\,(y)_m}{r!(1+m-r)!(x-r)}
(z+r)^{l+1}.\label{ker50}
\eeq
Let us separate the term at $r=0$ and calculate it using \eqref{ker45}. We obtain
\beq
\sum_{m=0}^{l-1}\frac{(y)_m}{(1+m)!}z^{l+1}=\frac{z^{l+1}}{(1-y)}+\frac{z^{l+1}(y)_l}{l!\,(y-1)}.
\label{ker51}
\eeq
In the remaining sum we substitute $r=s+1$ and interchange summations in $m$ and $s$, giving
\begin{align}
\sum_{m=0}^{l-1}&\sum_{s=0}^m\frac{x\,(-1)^{s+1}(y)_m(z+1+s)^{l+1}}{(m-s)!(s+1)!(x-s-1)}=
\sum_{s=0}^{l-1}\sum_{m=s}^{l-1}\frac{x\,(-1)^{s+1}(y)_m(z+1+s)^{l+1}}{(m-s)!(s+1)!(x-s-1)}
\nonumber\\
=x\,(y)_l&\sum_{s=0}^{l-1} (-1)^{s+1}\frac{(z+1+s)^{l+1}}{(l-s-1)!(s+1)!(y+s)(x-s-1)}
\nonumber\\
=x\,(y)_l&\sum_{t=0}^{l} (-1)^{t}\frac{(z+t)^{l+1}}{t!\,(l-t)!\,(y+t-1)(x-t)}-
\frac{z^{l+1}(y)_l}{l!\,(y-1)},
 \label{ker52}
\end{align}
where we used \eqref{ker44} to calculate the sum in $m$, set $s=t-1$ and subtracted the term with
$t=0$. Splitting the factors in the denominator as
\beq
\frac{1}{(y+t-1)(x-t)}=\frac{1}{(x+y-1)(x-t)}+\frac{1}{(x+y-1)(y+t-1)}, \label{ker53}
\eeq
we can finally evaluate  \eqref{ker52} using \eqref{ker40a1}.
The result reads
\beq
\frac{x\,(y)_l(x+z)^{l+1}}{(1-x-y)(-x)_{l+1}}+\frac{x(1-y+z)^{l+1}}{(1-y)(1-x-y)}-
\frac{z^{l+1}(y)_l}{l!\,(y-1)}.\label{ker54}
\eeq
Combining \eqref{ker49}, \eqref{ker51} and \eqref{ker54} we get the RHS of \eqref{ker40}.
\end{proof}

Now we can give the proof of  {\bf Theorem \ref{Th1}}, based on the above identities, and the
algebraic properties of $P_n(x)$ and $Q_n(x)$ from Section \ref{S33}.

\bigskip
\noindent {\it Proof of Theorem \ref{Th1}}. \:
From \eqref{ker9a} and \eqref{ker9b} we have
\begin{align}
&xP_{m}(x)Q_m(y)=P_{m+1}(x)Q_m(y)+\sum_{k=0}^M a_{k,m} P_{m-k}(x)Q_m(y),\label{ker34}\\
&yP_m(x) Q_{m}(y)=P_m(x)Q_{m-1}(y)+\sum_{k=0}^M a_{k,m+k}P_m(x)Q_{m+k}(y). \label{ker35}
\end{align}
Subtracting these two relations and summing over $m$ from $0$ to $n-1$,
we get after simplifications
\beq
(x-y)K_n^M(x,y)=P_n(x)Q_{n-1}(y)-\sum_{k=1}^M\sum_{m=0}^{k-1}a_{k,n+m}
P_{n-k+m}(x)Q_{n+m}(y).\label{ker38}
\eeq
Here we used the fact that $a_{k,n}=0$ for $k>n$ as follows from \eqref{ker9c}
and $Q_{-1}(x)=0$ which can be seen from the integral representation \eqref{ker14}.

Using formulas (\ref{ker18}-\ref{ker19}, \ref{ker29}) and the explicit form \eqref{ker9c}
of the coefficients $a_{k,n}$ we can rewrite the RHS as a bilinear differential operator
acting on the product $P_n(x)Q_n(y)$. Expanding the result in the basis
of symmetric functions $e_{M-l}(\nu_1,\ldots,\nu_M)$, $l=0,\ldots,M$
 and comparing with \eqref{ker5}
we can reformulate the statement of the theorem as the  identity
\begin{align}
&\sum_{k=1}^M\sum_{m=0}^{k-1}\sum_{j=0}^{k+1}
\frac{(-1)^{j+m}}{j!(k+1-j)!}(\delta_x-n)_{k-m}(\delta_y+n+1)_m(n+m-j+1)^{l+1}=\nonumber\\
&\sum_{i=0}^{l}\delta_x^{l-i}(-\delta_y)^i-\sum_{i=0}^{l}n^{l-i}(-\delta_y)^i
,\quad l=0,\ldots,M.
\label{ker39}
\end{align}
This identity contains three independent variables $n$, $\delta_x$
and $\delta_y$. First we notice that all contributions to the sum with $k>l$ are equal to $0$.
This follows from the fact that for any fixed  $k>l$ and  $m=0,\ldots,k-1$, the sum over $j$
is equal to $0$ because of the identity
$$
\sum_{i=0}^n (-1)^i { i^p \over i! (n - i)!} =
\sum_{i=0}^n (-1)^i {\delta_z^p(z^i) \over i! (n-i)!} \Big |_{z=1} =
{1 \over n!} \delta_z^p[(1-z)^n] \Big |_{z=1} = 0,
$$
valid for $0 \le p < n$.
Restricting the summation in $k$ to $1,\ldots,l$ we can use Lemma \ref{lem2} to calculate the sum
in the LHS of \eqref{ker39}. We obtain
\beq
-\frac{n^{l+1}}{n+\delta_y}+\frac{\delta_x^{l+1}}{\delta_x+\delta_y}+
\frac{(\delta_x-n)(-\delta_y)^{l+1}}{(n+\delta_y)(\delta_x+\delta_y)}.\label{ker55}
\eeq
Summing up geometric series in the RHS of \eqref{ker39} we get the same result, thus
verifying (\ref{ker39}) and establishing the Theorem. \hfill $\square$

\nsection{Generalization of Tracy and Widom theory}

The integrable form of the finite $n$ kernel (\ref{repr})  enables  the derivation
of a set of partial differential equations for
the gap probabilities.
We closely follow Strahov's \cite{St14} generalization of the original
approach of Tracy and Widom \cite{TW94c}
in his derivation of
analogous equations in relation to the hard edge kernel in integrable form  (\ref{2.2}), but
 with modifications
due to the final value of $N_0=n$. We remark that the approach of Strahov has also been
applied in \cite{Zh17} to study the gap probability in the hard edge scaled Muttalib--Borodin model.

Using the definitions (\ref{ker6}) of the
functions $\phi_j^{(n)}$, $\psi_j^{(n)}$, and (\ref{edge5}) of $\alpha_i$ we can derive
\beq
\delta_x\phi_j(x)=-\phi_{j+1}(x),\quad j=0,\ldots,M-1,\label{ker66}
\eeq
\beq
\delta_x\phi_M(x)=\sum_{k=1}^M(-1)^{M-k+1}e_{M-k+1}({\boldsymbol\nu})
\phi_k(x)+(-1)^{M-1}x(\phi_1(x)+n\phi_0(x)),\label{ker67}
\eeq
\begin{align}
&\delta_x\psi_0(x)=nx(-1)^M\psi_M(x),\label{ker68}
\\
&\delta_x\psi_j(x)=\psi_{j-1}(x)+(-1)^{M-j}(e_{M-j+1}({\boldsymbol\nu})-x\delta_{j,1})\psi_M(x),\quad
j=1,\ldots,M.\label{ker69}
\end{align}

Denoting the characteristic function of the interval $J$ \eqref{J}
by $\chi_J(x)$, we introduce the operator ${K}_{M,n}$  on $L^2(0,\infty)$ with the kernel
$K_n^M(x,y)\chi_J(y)$ which we denote as
\beq
K_{M,n}(x,y)=K_n^M(x,y)\chi_J(y).\label{ker62a}
\eeq
Let us notice that to restore a dependence on the parameter $\lambda$ entering \eqref{a2} all we need
is to make a substitution
\beq
Q_n(x)\rightarrow {\lambda}Q_n(x) \label{renorm}
\eeq
in all formulas in subsequent sections. This will only effect initial conditions for primary
variables satisfying equations (\ref{ker118}-\ref{ker124}) below.
For simplicity we set $\lambda=1$ and restore a dependence on $\lambda$ in Sections 9 and 10
when we analyze cases $M=1,2$ in details.

Similarly define the operators ${K}'_{M,n}$ and ${K}^T_{M,n}$ with kernels
\beq
K_{M,n}'(x,y):=K_n^M(y,x)\chi_J(y),\quad K_{M,n}^T(x,y):=K^M_n(y,x)\chi_J(x). \label{ker62b}
\eeq
We also define operators
\begin{align}
\rho_{M,n}=({1}-{K}_{M,n})^{-1},&\quad
{R}_{M,n}=({1}-{K}_{M,n})^{-1}{K}_{M,n}
\label{ker63},\\
\rho'_{M,n}=({1}-{K}'_{M,n})^{-1},&\quad
{R}'_{M,n}=({1}-{K}'_{M,n})^{-1}{K}'_{M,n}
\label{ker64},\\
{\rho}^T_{M,n}=({1}-{K}^T_{M,n})^{-1},&\quad
{R}^T_{M,n}=({1}-{K}^T_{M,n})^{-1}{K}^T_{M,n},
\label{ker65}
\end{align}
as well as the functions
\begin{align}
&Q_j^{(n)}(x;{\bf a})=(1-{K}_{M,n})^{-1} \phi_j^{(n)}(x),\label{ker71}\\
&P_j^{(n)}(x;{\bf a})=(1-{K}'_{M,n})^{-1}\psi_j^{(n)}(x),\quad
{\bf a}\equiv(a_1,\ldots,a_{2m}),\label{ker72}
\end{align}
\beq
V_{i,j}^{(n)}({\bf a})=\int_J\phi_i^{(n)}(x)P_j^{(n)}(x;{\bf a})dx=
\int_0^\infty \phi_i^{(n)}(x)P_j^{(n)}(x;{\bf a})\chi_J(x)dx \label{ker73}
\eeq
with $0\leq i,j\leq M$. As a final preliminary, we note
that by substituting  $P_{n-1}(x)$ from \eqref{ker10} in (\ref{edge6}) gives
\beq
K_n^M(x,y)=\left(n-{\delta_x}\right)\int_0^1 P_{n}(u x)Q_n(u y)du. \label{ker37}
\eeq

\begin{prop}\label{Th2}
For $ j=0,\ldots,M-1$ the functions $Q_j^{(n)}(x;{\bf a})$ satisfy the system of partial
differential equations
\begin{align}
\delta_x Q_j^{(n)}(x;{\bf a})&=(-1)^{M+1}\left[nQ_0^{(n)}(x;{\bf a})+
Q_1^{(n)}(x;{\bf a})\right]V_{j,M}^{(n)}({\bf a})-
Q_{j+1}^{(n)}(x;{\bf a})\nonumber\\
&-\sum_{k=1}^{2m}(-1)^k a_k R_{M,n}(x;a_k)Q_j^{(n)}(a_k;{\bf a}),\label{ker80}
\end{align}
while for $j=M$
\begin{align}
\delta_x Q_M^{(n)}(x;{\bf a})&=(-1)^{M+1}\left[nQ_0^{(n)}(x;{\bf a})+
Q_1^{(n)}(x;{\bf a})\right](x+V_{M,M}^{(n)}({\bf a}))+\nonumber\\
&+(-1)^M\sum_{k=0}^MQ_k^{(n)}(x;{\bf a})\left[nV_{0,k}^{(n)}({\bf a})+V_{1,k}^{(n)}({\bf a})
-(-1)^ke_{M+1-k}({\boldsymbol\nu})\right]\nonumber\\
&-\sum_{k=1}^{2m}(-1)^k a_k R_{M,n}(x,a_k)Q_M^{(n)}(a_k;{\bf a}).\label{ker81}
\end{align}
For $1\leq k\leq 2m$, $0\leq j\leq M$ we also have
\beq
\frac{\partial}{\partial a_k}Q_j^{(n)}(x;{\bf a})=(-1)^k R_{M,n}(x;a_k)Q_j^{(n)}(a_k;{\bf a}).
\label{ker82}
\eeq
\end{prop}

\begin{prop}\label{Th3}
For $j=2,\ldots,M$ the functions $P_j^{(n)}(x;{\bf a})$ satisfy the system of partial
differential equations
\begin{align}
\delta_x P_j^{(n)}(x;{\bf a})&=
(-1)^{M+1}P_M^{(n)}(x;{\bf a})\left[nV_{0,j}^{(n)}({\bf a})+V_{1,j}^{(n)}({\bf a})-
(-1)^je_{M-j+1}(\boldsymbol{\nu})\right]
+\nonumber\\
&+P_{j-1}^{(n)}(x;{\bf a})-\sum_{k=1}^{2m}(-1)^k a_k R'_{M,n}
(x,a_k)P_j^{(n)}(a_k;{\bf a}), \label{ker99}
\end{align}
while for $j=1$
\begin{align}
\delta_x P_1^{(n)}(x;{\bf a})&=
(-1)^{M+1}P_M^{(n)}(x;{\bf a})\left[nV_{0,1}^{(n)}({\bf a})+V_{1,1}^{(n)}({\bf a})+
e_{M}(\boldsymbol{\nu})-x\right]
+\nonumber\\
&+P_{0}^{(n)}(x;{\bf a})+(-1)^M\sum_{k=0}^MP_k^{(n)}(x;{\bf a})
V_{k,M}^{(n)}({\bf a})\nonumber\\
&-\sum_{k=1}^{2m}(-1)^k a_k R'_{M,n}
(x,a_k)P_1^{(n)}(a_k;{\bf a}),\label{ker100}
\end{align}
and for $j=0$
\begin{align}
\delta_x P_0^{(n)}(x;{\bf a})&=
(-1)^{M+1}P_M^{(n)}(x;{\bf a})\left[nV_{0,0}^{(n)}({\bf a})+V_{1,0}^{(n)}({\bf a})\right]
\nonumber\\
&+(-1)^M n x P_M^{(n)}(x;{\bf a})
+(-1)^Mn\sum_{k=0}^M P_k^{(n)}(x;{\bf a}) V_{k,M}^{(n)}({\bf a}) \nonumber\\
&-\sum_{k=1}^{2m}(-1)^k a_k R'_{M,n}
(x,a_k)P_0^{(n)}(a_k;{\bf a}).\label{ker101}
\end{align}
For $1\leq k\leq 2m$, $0\leq j\leq M$ we also have
\beq
\frac{\partial}{\partial a_k}P_j^{(n)}(x;{\bf a})=(-1)^k R'_{M,n}(x;a_k)P_j^{(n)}(a_k;{\bf a}),
\label{ker102}
\eeq
and for $1\leq l\leq 2m$, $0\leq j\leq M$
\beq
\frac{\partial}{\partial a_l}V_{i,j}^{(n)}({\bf a})=
(-1)^l Q_i^{(n)}(a_l;{\bf a})P_j^{(n)}(a_l;{\bf a}).\label{ker103}
\eeq
\end{prop}

These two propositions generalise  Propositions 4.1 and 4.2 \cite{St14}.

\nsection{Proofs}

Here we give proofs of  {\bf Propositions  \ref{Th2}, \ref{Th3}}.
To make it more structural we split the derivation of the partial differential equations
for the functions (\ref{ker71}-\ref{ker73})
into several steps.

First, we notice that for any two operators ${K}$ and ${L}$ (see e.g.~\cite[Prop.~9.3.4]{Forrester2010})
\beq
[{L},({1}-{K})^{-1}]=({1}-{K})^{-1}[{L},{K}]({1}-{K})^{-1}, \label{ker70a}
\eeq
and
\beq
\frac{d}{da}({1}-{K})^{-1}=({1}-{K})^{-1}\frac{d{K}}{da}({1}-{K})^{-1},\label{ker70b}
\eeq
where we imply that the operator $K$ smoothly depends on a parameter $a$. And with $J$ given by
(\ref{J}) we have
\beq\label{JJ}
{\partial \over \partial x} \chi_J(x) = \sum_{k=1}^{2m} (-1)^{k-1} \delta(x - a_k), \quad
{\partial \over \partial a_k} \chi_J(x) =  (-1)^{k} \delta(x - a_k).
\eeq

It is convenient to use the following notations
\beq
{D}\phi(x)=\frac{d}{dx}\phi(x),\quad {M}\phi(x)=x\phi(x). \label{ker70c}
\eeq
Obviously, the operator $\delta$ introduced earlier in \eqref{edge3} is equal to $\delta={MD}$.
We also notice that for any operator ${L}$ with a kernel $L(x,y)$
\begin{align}
&\int_\Omega (\delta_x L(x,y)-L(x,y)\delta_y) f(y) dy =
\int_\Omega (x\frac{d}{dx} L(x,y)-L(x,y)y\frac{d}{dy}) f(y) dy = \nonumber\\
&-L(x,y)yf(y)|_{y\in\p\Omega}+
(\delta_x+1)\int_\Omega L(x,y)f(y)dy +
\int_\Omega f(y)y\frac{\partial }{\partial y}L(x,y)dy.
\label{ker70d}
\end{align}
So if ${L}$ has a compact support, we have the identity for the kernel of the commutator
$[\delta,{L}]$
\beq
[\delta,{L}](x,y)=(\delta_x+\delta_y+1)L(x,y).\label{ker70e}
\eeq

\begin{lemma}\label{lem3}

The kernel of the operator $[\delta,(1-K_{M,n})^{-1}]$ is given by
\begin{align}
[\delta,(1-K_{M,n})^{-1}](x,y)&=
(-1)^{M+1}\{nQ_0^{(n)}(x;{\bf a})+Q_1^{(n)}(x;{\bf a})\}P_M^{(n)}(y;{\bf a})\chi_J(y)\nonumber\\
&-
\sum_{k=1}^{2m}(-1)^k a_k R_{M,n}(x,a_k)\rho_{M,n}(a_k,y). \label{ker74}
\end{align}
Similarly,
\begin{align}
[\delta,(1-K'_{M,n})^{-1}](x,y)&=
(-1)^{M+1}\{nQ_0^{(n)}(y;{\bf a})+Q_1^{(n)}(y;{\bf a})\}P_M^{(n)}(x;{\bf a})\chi_J(y)\nonumber\\
&-
\sum_{k=1}^{2m}(-1)^k a_k R'_{M,n}(x,a_k)\rho'_{M,n}(a_k,y). \label{ker75}
\end{align}
\end{lemma}

\begin{proof}
We start with the proof of \eqref{ker74}. According to \eqref{ker70e} we obtain
for the operator $K_{M,n}$
\begin{align}
&[\delta, K_{M,n}](x,y)=(\delta_x+\delta_y+1)(n-\delta_x)\left\{\chi_J(y)
\int_{0}^1P_n(ux)Q_n(uy)du\right\}\nonumber\\
&=K_n^M(x,y)\left\{\chi_J(y)+y\frac{\p}{\p y}(\chi_J(y))\right\}+
\chi_J(y)(n-\delta_x)\int_{0}^1u\frac{\p}{\p u}
(P_n(ux)Q_n(uy))du=\nonumber\\
&=K_n^M(x,y)y\frac{\p}{\p y}(\chi_J(y))+
(n-\delta_x)P_n(x)Q_n(y)\chi_J(y)=\nonumber\\
&=(-1)^{M+1}\left\{n\phi_0^{(n)}(x)+\phi_1^{(n)}(x)\right\}\psi_M^{(n)}(y)\chi_J(y)
-\sum_{k=1}^{2m}(-1)^{k}a_k K_n^M(x,a_k)\delta(y-a_k), \label{ker76}
\end{align}
where we used  \eqref{JJ}, \eqref{ker37} and (\ref{repr}-\ref{ker6}) to express $P_n(x)$ and $Q_n(y)$ in terms
of $\phi$'s and $\psi$'s.

Now
\begin{align}
&[\delta,(1-K_{M,n})^{-1}](x,y)=\int_0^\infty du \>\rho_{M,n}(x,u)
\int_0^\infty dv\> \rho_{M,n}(v,y)\>[\delta, K_{M,n}](u,v)=\label{ker77}\\
&(-1)^{M+1}\int_0^\infty du\>\rho_{M,n}(x,u)\left\{n\phi_0^{(n)}(u)+\phi_1^{(n)}(u)\right\}
\int_0^\infty dv \rho_{M,n}(v,y)\psi_M^{(n)}(v)\chi_J(v)\nonumber\\
&-\sum_{k=1}^{2m}(-1)^{k}a_k\int_0^\infty du \>\rho_{M,n}(x,u)K_n^M(u,a_k)
\int_0^\infty dv \rho_{M,n}(v,y)\delta(v-a_k)=\nonumber\\
&(-1)^{M+1}\{nQ_0^{(n)}(x;{\bf a})+Q_1^{(n)}(x;{\bf a})\}P_M^{(n)}(y;{\bf a})\chi_J(y)-
\sum_{k=1}^{2m}(-1)^k a_k R_{M,n}(x,a_k)\rho_{M,n}(a_k,y), \nonumber
\end{align}
where we used (\ref{ker62a}-\ref{ker65}), (\ref{ker71}-\ref{ker72}) and
\beq
\rho_{M,n}(y,x)\chi_J(y)=\chi_J(x)\rho'_{M,n}(x,y).\label{ker78}
\eeq
For the kernel of $[\delta, K'_{M,n}]$ we obtain in the same way
\begin{align}
[\delta, K'_{M,n}](x,y)
&=(-1)^{M+1}\left\{n\phi_0^{(n)}(y)+\phi_1^{(n)}(y)\right\}\psi_M^{(n)}(x)\chi_J(y)\nonumber\\
&-\sum_{k=1}^{2m}(-1)^{k}a_k K_n^M(a_k,x)\delta(y-a_k) \label{ker79}
\end{align}
and \eqref{ker75} follows by  calculation similar to \eqref{ker77}.
\end{proof}

\subsection{\bf Proof of Proposition {\ref{Th2}}.}
We have
\begin{align}
&\delta_x Q_j^{(n)}(x;{\bf a})=\delta_x\{(1-K_{M,n})^{-1}\phi_j^{(n)}(x)\}=\nn \\
&[\delta,(1-K_{M,n})^{-1}]\phi_j^{(n)}(x)+(1-K_{M,n})^{-1}\delta \phi_j^{(n)}(x).\label{ker83}
\end{align}
For $j=0,\ldots,M-1$ we obtain from \eqref{ker6}
\beq
\delta \phi_j^{(n)}(x)=-\phi_{j+1}^{(n)}(x)\label{ker84}
\eeq
and \eqref{ker80} immediately follow by applying \eqref{ker74} from Lemma \ref{lem3} to the RHS
of \eqref{ker83}.

Now
\beq
\delta_x Q_M^{(n)}(x;{\bf a})=
[\delta,(1-K_{M,n})^{-1}]\phi_M^{(n)}(x)+(1-K_{M,n})^{-1}\delta\phi_M^{(n)}(x).\label{ker85}
\eeq
Using \eqref{ker74} we obtain for the first term in \eqref{ker85}
\begin{align}
[\delta,(1-K_{M,n})^{-1}]\phi_M^{(n)}(x)&=
(-1)^{M+1}\left[nQ_0^{(n)}(x;{\bf a})+
Q_1^{(n)}(x;{\bf a})\right]V_{M,M}^{(n)}({\bf a})\nonumber\\
&-\sum_{k=1}^{2m}(-1)^k a_k R_{M,n}(x,a_k)Q_M^{(n)}(a_k;{\bf a}).\label{ker86}
\end{align}

We can use  \eqref{ker67} to calculate $\delta\phi_M^{(n)}(x)$ in terms
of $\delta\phi_j^{(n)}(x)$, $j=0,\ldots,M$
\begin{align}
(1-K_{M,n})^{-1}\delta\phi_M^{(n)}(x)&=
\sum_{k=1}^M(-1)^{M-k+1}e_{M-k+1}({\boldsymbol\nu})
Q_k^{(n)}(x;{\bf a})\nonumber\\
&+(-1)^{M-1}(1-K_{M,n})^{-1}M \Big ( \phi_1(x)+n\phi_0(x)\Big ).\label{ker87}
\end{align}
Next we require
\beq
(1-K_{M,n})^{-1}M\phi_j(x)=[(1-K_{M,n})^{-1},M]\phi_j(x)+xQ_j^{(n)}(x;{\bf a})\label{ker88}
\eeq
for $j=0,1$.
Using (\ref{repr}-\ref{ker5}) we obtain
\beq
[K_{M,n},M](x,y)=-\sum_{k=0}^M\phi_k(x)\psi_k(y)\chi_J(y) \label{ker89}
\eeq
and
\beq
[(1-K_{M,n})^{-1},M](x,y)=-\sum_{k=0}^M Q_k^{(n)}(x;{\bf a})P_k^{(n)}(y;{\bf a})\chi_J(y).
\label{ker90}
\eeq
As a result
\begin{align}
(1-K_{M,n})^{-1}M\{\phi_1(x)+n\phi_0(x)\}&=
x(nQ_0^{(n)}(x;{\bf a})+Q_1^{(n)}(x;{\bf a}))\nonumber\\
&-
\sum_{k=0}^M Q_k^{(n)}(x;{\bf a})\Big (
nV_{0,k}^{(n)}({\bf a})+V_{1,k}^{(n)}({\bf a}) \Big ).\label{ker91}
\end{align}
Combining (\ref{ker85}-\ref{ker87}) and \eqref{ker91} we obtain \eqref{ker81}.

It remains to prove \eqref{ker82}. From \eqref{JJ} we derive
\beq
\frac{\p}{\p a_k}\Big ( K_n^M(x,y)\chi_J(y) \Big )=
(-1)^k K_n^M(x,a_k)\delta(y-a_k).\label{ker92}
\eeq
Since
\beq
\frac{\p}{\p a_k}(1-K_{M,n})^{-1}=(1-K_{M,n})^{-1}\frac{\p}{\p a_k}K_{M,n}(1-K_{M,n})^{-1},\label{ker93}
\eeq
we get similar to the calculation in \eqref{ker77}
\beq
\frac{\p}{\p a_k}(1-K_{M,n})^{-1}(x,y)=(-1)^k R_{M,n}(x,a_k)\rho_{M,n}(a_k,y)\label{ker94}
\eeq
and as a result
\beq
\frac{\partial}{\partial a_k}Q_j^{(n)}(x;{\bf a})=
\frac{\p}{\p a_k}(1-K_{M,n})^{-1}\phi_j(x)=
(-1)^k R_{M,n}(x,a_k)Q_j^{(n)}(a_k;{\bf a}).\label{ker95}
\eeq
\hfill\qedsymbol

\subsection{\bf Proof of Proposition {\ref{Th3}}.}

We have
\begin{align}
\delta_x P_j^{(n)}(x;{\bf a})=
[\delta, (1-K'_{M,n})^{-1}]\psi_j(x)+(1-K'_{M,n})^{-1}\delta\psi_j(x).
\label{ker104}
\end{align}
We evaluate the first term on the RHS using \eqref{ker75} from Lemma \ref{lem3}.

Noticing that
\begin{align}
&\int_J Q_i^{(n)}(x;{\bf a})\psi_j(x)dx=
\int_{0}^\infty\left\{\int_0^\infty \rho_{M,n}(x,y)\phi_i(y)dy\right\}
\psi_i(x)\chi_J(x)dx=\nonumber\\
&\int_0^\infty \chi_J(y)\phi_i(y)dy
\int_{0}^\infty\rho'_{M,n}(y,x)\psi_i(x)dx=
\int_J\phi_i(y)P_j^{(n)}(y;{\bf a})dy=V_{i,j}^{(n)}({\bf a}),\label{ker105}
\end{align}
we obtain
\begin{align}
[\delta, (1-K'_{M,n})^{-1}]\psi_j(x)&=
(-1)^{M+1}P_M^{(n)}(x;{\bf a})\left[nV_{0,j}^{(n)}({\bf a})+V_{1,j}^{(n)}({\bf a})\right]\nonumber\\
&-\sum_{k=1}^{2m}(-1)^k a_k R'_{M,n}
(x,a_k)P_j^{(n)}(a_k;{\bf a}).\label{ker106}
\end{align}

For the second term in \eqref{ker104} we consider three cases. For $2\leq j\leq M$
\begin{align}
(1-K'_{M,n})^{-1}\delta\psi_j(x)&=(1-K'_{M,n})^{-1}\{\psi_{j-1}(x)+(-1)^{M-j}e_{M-j+1}({\bf \nu})
\psi_M(x)\}\nonumber\\
&=P_{j-1}^{(n)}(x;{\bf a})+(-1)^{M-j}e_{M-j+1}({\bf \nu})P_{M}^{(n)}(x;{\bf a})
\label{ker107}
\end{align}
and this together with \eqref{ker106} gives \eqref{ker99}.
For $j=0,1$ we need to calculate the commutator $[(1-K'_{M,n})^{-1},M]$. The calculation
is similar to (\ref{ker89}-\ref{ker90}) and gives
\beq
[(1-K'_{M,n})^{-1},M](x,y)=\sum_{k=0}^M P_k^{(n)}(x;{\bf a})Q_k^{(n)}(y;{\bf a})\chi_J(y),
\label{ker108}
\eeq
implying
\beq
(1-K'_{M,n})^{-1}M \psi_{j}(x)=x P_j^{(n)}(x;{\bf a})+\sum_{k=0}^M P_k^{(n)}(x;{\bf a})
V_{k,j}^{(n)}({\bf a}).\label{ker109}
\eeq
Now for $j=1$ we obtain
\begin{align}
&(1-K'_{M,n})^{-1}\delta\psi_1(x)=1-K'_{M,n})^{-1}\left\{\psi_{0}(x)
+(-1)^M \left(x-e_{M}({\bf \nu})\right)\psi_M(x)\right\}\nonumber\\
&=P_{0}^{(n)}(x;{\bf a})+
(-1)^M(x-e_{M}({\bf \nu}))P_M^{(n)}(x;{\bf a})+(-1)^M\sum_{k=0}^M P_k^{(n)}(x;{\bf a})
V_{k,M}^{(n)}({\bf a}).\label{ker110}
\end{align}
Combining (\ref{ker104}, \ref{ker106}, \ref{ker110}) we come to \eqref{ker100}.
Finally, for $j=0$ we use \eqref{ker68}
\begin{align}
&(1-K'_{M,n})^{-1}\delta\psi_0(x)=(1-K'_{M,n})^{-1}
\left\{(-1)^M n x\psi_M(x)\right\}
\nonumber\\
&=(-1)^M nx P_M^{(n)}(x;{\bf a})
+(-1)^Mn\sum_{k=0}^M P_k^{(n)}(x;{\bf a}) V_{k,M}^{(n)}({\bf a})
\label{ker111}
\end{align}
and we obtain \eqref{ker101}.

It remains to check (\ref{ker102}-\ref{ker103}).
We have
\beq
\frac{\p}{\p a_k}K_{M,n}'(x,y)=(-1)^k K(a_k,x)\delta(a_k-y), \label{ker112}
\eeq
\begin{align}
\frac{\p}{\p a_k}(1-K_{M,n}')^{-1}(x,y)&=(-1)^k
\left\{\int_0^\infty \rho_{M,n}'(x,u)K(a_k,u)\right\} \rho'_{M,n}(a_k,y)\nonumber\\
&=(-1)^kR_{M,n}'(x,a_k)\rho'_{M,n}(a_k,y). \label{ker112a}
\end{align}
and \eqref{ker102} follows.
Similarly,
\begin{align}
&\frac{\partial}{\partial a_l}V_{i,j}^{(n)}({\bf a})=
\int_0^\infty \phi_i(x)\frac{\partial}{\partial a_l}\left\{
P_j^{(n)}(x;{\bf a})\chi_J(x)\right\}\nonumber\\
&=(-1)^l\phi_i(a_l)P_j^{(n)}(a_l;{\bf a})+
(-1)^l\left\{\int_0^\infty \phi_i(x)\chi_J(x)R'_{M,n}(x,a_l)\right\}P_j^{(n)}(a_l;{\bf a})\nonumber\\
&=(-1)^l\left\{\int_0^\infty \phi_i(x)\chi_J(x)\left(\delta(x-a_l)+R'_{M,n}(x,a_l)\right)
\right\}P_j^{(n)}(a_l;{\bf a})\nonumber\\
&=(-1)^l\left\{\int_0^\infty \phi_i(x)\chi_J(x)\rho'_{M,n}(x,a_l)
\right\}P_j^{(n)}(a_l;{\bf a})=(-1)^l Q_i^{(n)}(a_l;{\bf a})P_j^{(n)}(a_l;{\bf a}).\label{ker113}
\end{align}
\hfill\qedsymbol

\nsection{Kernels of resolvent operators}

From {\bf Theorem \ref{Th1}} the kernel $K_{M,n}$ and its transpose $K_{M,n}'$ in integrable form are
\beq
K_n^M(x,y)\chi_J(y)=\frac{\ds\sum_{j=0}^M \phi_j(x)\psi_j(y)}{x-y}\chi_J(y),  \quad
K_n^M(y,x)\chi_J(y)=-\frac{\ds\sum_{j=0}^M \psi_j(x)\phi_j(y)}{x-y}\chi_J(y). \label{res2}
\eeq

The following proposition is the direct analog of Proposition 4.4  from \cite{St14}
\begin{prop}
The kernels of the resolvent operators $R_{M,n}$ and $R'_{M,n}$ (\ref{ker63}-\ref{ker64})
are given by
\beq
R_{M,n}(x,y)=\frac{\ds \sum_{j=0}^M Q_j^{(n)}(x;{\bf a})
P_j^{(n)}(y;{\bf a})}{x-y}\chi_J(y), \label{res3}
\eeq
\beq
R'_{M,n}(x,y)=-\frac{\ds \sum_{j=0}^M P_j^{(n)}(x;{\bf a})
Q_j^{(n)}(y;{\bf a})}{x-y}\chi_J(y). \label{res4}
\eeq
\end{prop}

\begin{proof}
Let us first calculate the kernel of the operator $R_{M,n}(x,y)$,
\begin{align}
&(x-y)R_{M,n}(x,y)=[M,R_{M,n}](x,y)=[M,-1+(1-K_{M,n})^{-1}](x,y)=\nn\\
&-[(1-K_{M,n})^{-1},M](x,y)=
\sum_{j=0}^M Q_j^{(n)}(x;{\bf a})P_j^{(n)}(y;{\bf a})\chi_J(y),\label{res5}
\end{align}
where we used the fact that $R_{M,n}=(1-K_{M,n})^{-1}K_{M,n}=-1+(1-K_{M,n})^{-1}$ and \eqref{ker90}.

Repeating these calculations for the operator $R'_{M,n}$ and using \eqref{ker108} we come to \eqref{res4}.
\end{proof}

Another important property of the kernel $R_{M,n}$ generalises Proposition 4.6 of \cite{St14}.
\begin{prop}\label{prop5}
The kernel $R_{M,n}(x,y)$ satisfies the partial differential equation
\beq
\left(\delta_x+\delta_y+\sum_{k=1}^{2m}\delta_{a_k}+1\right)R_{M,n}(x,y)=
(-1)^{M+1}\{nQ_0^{(n)}(x;{\bf a})+Q_1^{(n)}(x;{\bf a})\}P_M^{(n)}(y;{\bf a})\chi_J(y).\label{res6}
\eeq

\end{prop}

\begin{proof}
By \eqref{ker70e} and {\bf Lemma \ref{lem3}} we have
\begin{align}
&\left(\delta_x+\delta_y+1\right)R_{M,n}(x,y)=[\delta,R_{M,n}(x,y)]=
[\delta,(1-K_{M,n})^{-1}](x,y)=\nn\\
&(-1)^{M+1}\{nQ_0^{(n)}(x;{\bf a})+Q_1^{(n)}(x;{\bf a})\}P_M^{(n)}(y;{\bf a})\chi_J(y)-
\sum_{k=1}^{2m}(-1)^k a_k R_{M,n}(x,a_k)\rho_{M,n}(a_k,y). \label{res7}
\end{align}
Finally, from  \eqref{ker94} we have
\beq
\sum_{k=1}^{2m}a_k\frac{\partial}{\partial a_k}
R_{M,n}(x,y)=\sum_{k=1}^{2m}(-1)^k a_k R_{M,n}(x,a_k)\rho_{M,n}(a_k,y) \label{res8}
\eeq
and \eqref{res6} follows.

\end{proof}

\nsection{Strahov's equations for primary variables}

We now define analogs of Strahov primary variables for finite $n$
for $j=0,\ldots,M$, $k=1,\ldots,2m$
\beq
x_{j,k}^{(n)}=\ep_k(1-K_{M,n})^{-1}\phi^{(n)}_j(a_k),\quad
y_{j,k}^{(n)}=\ep_k(1-K'_{M,n})^{-1}\psi^{(n)}_j(a_k),\label{ker114}
\eeq
where $\ep_k$ are some constants. It is easy to see that they enter  the equations
in  Theorems \ref{Th2}, \ref{Th3} only as $\ep_k^2$ and it is convenient to choose
\beq
\ep_k^2=(-1)^{k+1}.\label{ker115}\eeq
We realise this by following the choice of \cite{St14},
\beq
\ep_{2k}=i,\quad \ep_{2k-1}=1, \quad k=1,\ldots,m. \label{ker116}
\eeq

Let also define variables for $j=0,\ldots,M$
\begin{align}
&\xi^{(n)}_j=(-1)^M\left\{ n V_{0,j}^{(n)}({\bf a})+V_{1,j}^{(n)}({\bf a})-(-1)^{j}
e_{M+1-j}(\boldsymbol{\nu})\right\},\quad \boldsymbol{\nu}=(\nu_0,\ldots,\nu_M),\nonumber\\
&\eta_j^{(n)}=(-1)^MV_{j,M}^{(n)}({\bf a}).\label{ker117}
\end{align}

\begin{theorem} \label{Th4}
The functions $x_{j,k}^{(n)}$, $y_{j,k}^{(n)}$, $\xi^{(n)}_j$ and $\eta_j^{(n)}$
satisfy systems of partial differential equations
\begin{enumerate}
\item For $j=0,\ldots,M-1,$ $l=1,\ldots,2m$
\beq
a_l\frac{\p x_{j,l}^{(n)}}{\p a_l}=-\eta_j^{(n)}(nx_{0,l}^{(n)}+x_{1,l}^{(n)})-
x_{j+1,l}^{(n)}+\sum_{\substack{k=1\\k\neq l}}^{2m}
\frac{a_k }{a_l-a_k}x_{j,k}^{(n)}\sum_{i=0}^M x_{i,l}^{(n)}y_{i,k}^{(n)}\label{ker118}
\eeq
and for $j=M$, $l=1,\ldots,2m$
\beq
a_l\frac{\p x_{M,l}^{(n)}}{\p a_l}=-(\eta_M^{(n)}+(-1)^Ma_l)(nx_{0,l}^{(n)}+x_{1,l}^{(n)})
+
\sum_{i=0}^M x_{i,l}^{(n)}\Bigg\{\xi^{(n)}_i+
\sum_{\substack{k=1\\k\neq l}}^{2m}\frac{a_k }{a_l-a_k}x_{M,k}^{(n)}y_{i,k}^{(n)}\Bigg\}.
\label{ker119}
\eeq

\item For $j=0$, $l=1,\ldots,2m$
\beq
a_l\frac{\p y_{0,l}^{(n)}}{\p a_l}=\left[(-1)^Mna_l-\xi_0^{(n)}\right]y_{M,l}^{(n)}
+
\sum_{i=0}^M y_{i,l}^{(n)}\Bigg\{n\eta^{(n)}_i+
\sum_{\substack{k=1\\k\neq l}}^{2m}\frac{a_k }{a_k-a_l}x_{i,k}^{(n)}y_{0,k}^{(n)}\Bigg\},\label{ker120}
\eeq
for $j=1$, $l=1,\ldots,2m$
\beq
a_l\frac{\p y_{1,l}^{(n)}}{\p a_l}=\left[(-1)^Ma_l-\xi_1^{(n)}\right]y_{M,l}^{(n)}
+y_{0,l}^{(n)}+
\sum_{i=0}^M y_{i,l}^{(n)}\Bigg\{\eta^{(n)}_i+
\sum_{\substack{k=1\\k\neq l}}^{2m}\frac{a_k }{a_k-a_l}x_{i,k}^{(n)}y_{1,k}^{(n)}\Bigg\},\label{ker121}
\eeq
and for $j=2,\ldots,M$, $l=1,\ldots,2m$
\beq
a_l\frac{\p y_{j,l}^{(n)}}{\p a_l}=-\xi_j^{(n)}y_{M,l}^{(n)}
+y_{j-1,l}^{(n)}+
\sum_{\substack{k=1\\k\neq l}}^{2m}\frac{a_k }{a_k-a_l}y_{j,k}^{(n)}
\sum_{i=0}^M x_{i,k}^{(n)} y_{i,l}^{(n)}.\label{ker122}
\eeq
\item
For $j=0,\ldots,M$ and $k,l=1,\ldots,2m$, $k\neq l$
\beq
\frac{\p x_{j,l}^{(n)}}{\p a_k}=-\frac{x_{j,k}^{(n)}}{a_l-a_k}\sum_{i=0}^M
x_{i,l}^{(n)}y_{i,k}^{(n)},\quad
\frac{\p y_{j,l}^{(n)}}{\p a_k}=-\frac{y_{j,k}^{(n)}}{a_k-a_l}\sum_{i=0}^M
x_{i,k}^{(n)}y_{i,l}^{(n)}.
\label{ker123}
\eeq
\item
For $j=0,\ldots,M$, $l=1,\ldots,2m$
\beq
\frac{\p \xi_j^{(n)}}{\p a_l}=
(-1)^{M+1}(nx_{0,l}^{(n)}+x_{1,l}^{(n)})y_{j,l}^{(n)},\quad
\frac{\p \eta_j^{(n)}}{\p a_l}=
(-1)^{M+1}x_{j,l}^{(n)}y_{M,l}^{(n)}.\label{ker124}
\eeq
\end{enumerate}
\end{theorem}

The proof of this theorem is straightforward. We set $x=a_l$ in {\bf Propositions (\ref{Th2}-\ref{Th3})}
and in formulas for resolvent kernels (\ref{res3}-\ref{res4}).

Comparison with the corresponding result in \cite[Prop.~3.3]{St14} one sees that the
modification of Strahov equations to finite $n$ is very simple. It looks even simpler
at the level of symplectic structure. Consider a dynamical system with variables
$(x_{j,k}^{(n)},\xi_j^{(n)};y_{j,k}^{(n)},\eta_{j}^{(n)})$ and introduce
the same Poisson brackets as in \cite[eq. (3.24)]{St14},
\beq
\{x_{j,l}^{(n)},y_{i,k}^{(n)}\}=\frac{1}{a_k}\delta_{k,l}\delta_{i,j},\quad
\{\xi_j^{(n)},\eta_i^{(n)}\}=(-1)^M\delta_{i,j}.\label{ker125}
\eeq
with all remaining Poisson brackets equal to $0$.

\begin{theorem}\label{Th5}
The system of equations from the Theorem \ref{Th4} associated with the kernel $K_{M,n}$
can be written in Hamiltonian form
\begin{align}
\frac{\p x_{j,k}^{(n)}}{\p a_l}=\left\{x_{j,k}^{(n)},H_l^{(n)}\right\}&,\quad
\frac{\p y_{j,k}^{(n)}}{\p a_l}=\left\{y_{j,k}^{(n)},H_l^{(n)}\right\},\label{ker126a} \\
\frac{\p \xi_j^{(n)}}{\p a_l}=\left\{\xi_j^{(n)},H_l^{(n)}\right\}&,\quad
\frac{\p \eta_j^{(n)}}{\p a_l}=\left\{\eta_j^{(n)},H_l^{(n)}\right\},\label{ker126}
\end{align}
for $j=0,\ldots,M$, $k,l=1,\ldots,2m$. The Hamiltonians are given by
\begin{align}
&H_l^{(n)}=\Big\{(-1)^{M+1}a_ly_{M,l}^{(n)}-\sum_{i=0}^M\eta_i^{(n)}y_{i,l}^{(n)}\Big\}
\left(nx_{0,l}^{(n)}+x_{1,l}^{(n)}\right)-\nonumber\\
&-\sum_{j=0}^{M-1}x_{j+1,l}^{(n)}y_{j,l}^{(n)}+y_{M,l}^{(n)}\sum_{i=0}^M\xi_i^{(n)}
x_{i,l}^{(n)}+
\sum_{\substack{k=1\\k\neq l}}^{2m}\frac{a_k }{a_l-a_k}
\sum_{i,j=0}^M x_{i,k}^{(n)}x_{j,l}^{(n)} y_{i,l}^{(n)}y_{j,k}^{(n)}.\label{ker127}
\end{align}

The Hamiltonians $H_l^{(n)}$ are in involution
\beq
\{H_l^{(n)},H_r^{(n)}\}=0,
\label{ker127a}
\eeq
where $1\leq l,r\leq 2m$.

\end{theorem}

Using the Poisson brackets \eqref{ker125} it is easy to see that
equations  (\ref{ker118}-\ref{ker124}) follow from (\ref{ker126a},\ref{ker126}).
The relation \eqref{ker127a} is a tedious but straightforward calculation
with the use of \eqref{ker125}.

Comparing \eqref{ker127} with eq. (3.28) in \cite{St14} we see that
the only modification of the Hamiltonians
from the case $n\to\infty$ to finite $n$ reduces to the change
\beq
x_{0,l}\rightarrow nx_{0,l}^{(n)}+x_{1,l}^{(n)}.\label{ker128}
\eeq
in the first term of \eqref{ker127}.

\begin{prop}\label{Th6}
The Hamiltonians $H_l^{(n)}$ \eqref{ker127} can be written as
\beq
H_l^{(n)}=a_l\frac{\p}{\p a_l}\log\left(\det(1-K_{M,n})\right).\label{ker127b}
\eeq
\end{prop}

\begin{proof}
The proof is standard (see, for example, \cite[Ex.~9.3 q.1]{Forrester2010})
and based on calculation of the trace of the resolvent operator $R_{M,n}$.
Using \eqref{ker63} and \eqref{ker92}  we obtain
\beq
\left((1-K_{M,n})^{-1}\frac{\p K_{M,n}}{\p a_l}\right)(x,y)=
(-1)^lR_{M,n}(x,a_l)\delta(a_l-y)\label{ker127c}
\eeq
and
\beq
\frac{\p}{\p a_l}\log\left(\det(1-K_{M,n})\right)
=-\Tr\left((1-K_{M,n})^{-1}\frac{\p K_{M,n}}{\p a_l}\right)=
(-1)^{l+1} R_{M,n}(a_l,a_l).\label{ker127d}
\eeq
We can calculate $R_{M,n}(a_l,a_l)$ by taking the limit $x\to y=a_l$ in \eqref{res3}.
First notice that from continuity of the kernel $R_{M,n}(x,y)$  at $x=y$
we have
\beq
\sum_{j=0}^M P_{j}^{(n)}(x;{\bf a})\>
Q_{j}^{(n)}(x;{\bf a})=0\label{ker129b}
\eeq
and as a result we obtain
\beq
a_l R_{M,n}(a_l,a_l)=\sum_{j=0}^M \left.P_{j}^{(n)}(a_l;{\bf a})\> x\frac{\p}{\p x}
Q_{j}^{(n)}(x;{\bf a})\right|_{x=a_l}=(-1)^{l+1}H_l^{(n)}, \label{ker128a}
\eeq
where we used (\ref{ker80}-\ref{ker81}) to calculate the the derivatives
of $Q_{j}^{(n)}(x;{\bf a})$ at $x=a_l$ for $j=0,\ldots,m$
and the explicit expressions \eqref{ker127} for $H_l^{(n)}$.
Comparing \eqref{ker127d} and \eqref{ker128a} we
obtain
\eqref{ker127b}.
\end{proof}
We can now define a sequence of $\tau$-functions $\tau_n$, $n=1,2,\ldots$
\beq
\tau_n({\bf a})=\det(1-K_{M,n}),\label{ker130}
\eeq
and the closed form
\beq
w_n({\bf a})=d\log\tau_n({\bf a}). \label{ker131}
\eeq
We have
\beq
w_n({\bf a})=\sum_{l=1}^{2m}H_l^{(n)}\frac{da_l}{a_l}.\label{ker132}
\eeq

For the simplest case $m=1$, $J=(0,s)$, $a_1=0$, $a_2=s$ the system (\ref{ker118}-\ref{ker124})
becomes the system of nonlinear differential equations in $s$.
The gap probability $E_n^M(0;J)$ coincides with the tau-function $\tau_n(s)$
which is given
by
\beq
\tau_n(s)=\exp\left\{\int_0^s{\frac{dt}{t} H_n(t)}\right\}, \label{ker133}
\eeq
where $H_n(s)=H_2^{(n)}(s)$.
Using variables $x_j^{(n)}=x_{j,2}^{(n)}$, $y_j^{(n)}=y_{j,2}^{(n)}$
we can rewrite \eqref{ker127} for $l=2$ as
\beq
H_n(s)=(-1)^{M+1}s(nx_0+x_1)y_M-
\sum_{i=0}^{M-1}x_{i+1}^{(n)}y_i^{(n)}+\sum_{i=0}^M\left\{
y_M^{(n)}\xi_i^{(n)}x_i^{(n)}-(nx_0+x_1)\eta_i^{(n)}y_i^{(n)}\right\}.\label{ker134}
\eeq
Now {\bf Proposition \ref{prop5}} gives
\beq
(sR_{M,n}(s,s))'=(-1)^{M+1}\ep_2^2(nx_0+x_1)y_M=(-1)^M(nx_0+x_1)y_M.\label{ker135}
\eeq
Comparing this with the equations \eqref{ker124} we can integrate \eqref{ker135} and derive
\beq
sR_{M,n}(s,s)=-(n\eta_0(s)+\eta_1(s))+C_0,\label{ker136}
\eeq
where $C_0$ is the integration constant. Taking into account \eqref{ker73} and \eqref{ker117}
we obtain in the limit $s\to0$ $\eta_j(0)=0$, $j=0,\ldots,M$ and as a result $C_0=0$.
Comparing \eqref{ker136} with \eqref{ker128a} we finally get
an alternative expression for the Hamiltonian \eqref{ker134}
\beq
H_n(s)=n\eta_0(s)+\eta_1(s) \label{ker137}
\eeq
and the expression for the tau-function in terms of primary variables
\beq
\tau_n(s)=\exp\left\{\int_0^s{\frac{n\eta_0(t)+\eta_1(t)}{t}dt}\right\}.\label{ker138}
\eeq
The sequence of tau-functions $\tau_n(s)$ should satisfy Toda-type recurrence relations,
but we postpone investigation of this possibility to another occasion.

\nsection{Isomonodromic deformation}

Here we briefly discuss the isomonodromic deformation of the system given by
the {\bf Theorem \ref{Th4}}. Again this section is a straightforward generalization of
 results of Section 3.4 in \cite{St14}
to finite $n$.
Introduce a set of $(M+1)\times(M+1)$ matrices
\beq
E_n=(-1)^{M+1}
\left(\begin{array}{ccccc}
0& 0 & 0 & \ldots & 0 \\
0& 0 & 0 & \ldots & 0 \\
\vdots & \, & \ddots  & \, & \vdots \\
0& 0 & 0 & \ldots & 0 \\
n & 1 & 0 & \ldots & 0 \end{array}\right),\>\>
C_n=\left(\begin{array}{lllll}
-n\eta_0 & -\eta_0-1 & 0 & \ldots & 0 \\
-n\eta_1 & -\eta_1 & -1 & \ldots & 0\\
\phantom{aaa}\vdots & \phantom{aa}\vdots & \ddots  & \, & \phantom{,}\vdots \\
-n\eta_{M-1} & -\eta_{M-1} & 0 & \ldots & -1\\
-n\eta_M+\xi_0 & -\eta_M+\xi_1 & \xi_2 & \ldots & \xi_M
\end{array}\right)\label{Lax1}
\eeq
and a set of residue matrices
\beq
A^{(l)}_n=\left(\begin{array}{c}x_{0,l}^{(n)}\\
x_{1,l}^{(n)}\\
\vdots \\
x_{M,l}^{(n)}\end{array}\right)\otimes
(y_{0,l}^{(n)},\,y_{1,l}^{(n)},\,\ldots,\,y_{M,l}^{(n)}).\label{Lax2}
\eeq

The following proposition is the analogue of Proposition 3.6 from \cite{St14}
\begin{prop}\label{prop6}
The differential equations (\ref{ker118}-\ref{ker124}) can be rewritten in the matrix
form
\begin{align}
a_l\frac{\p}{\p a_l}A^{(l)}_n&=[C_n+a_l E_n,A^{(l)}_n]+\sum_{\substack{k=1\\k\neq l}}^{2m}
\frac{a_k}{a_l-a_k}[A^{(k)}_n,A^{(l)}_n], \quad l=1,\ldots,2m, \label{Lax3}\\
\frac{\p}{\p a_k}A^{(l)}_n&=\frac{[A^{(l)}_n,A^{(k)}_n]}{a_l-a_k},\quad
k\neq l=1,\ldots,2m,\label{Lax4}\\
\frac{\p}{\p a_l}C_n&=[E_n,A^{(l)}_n],\quad l=1,\ldots,2m.\label{Lax5}
\end{align}
The Hamiltonians $H_l^{(n)}$ have the form
\beq
H_l^{(n)}=\Tr(C_n A^{(l)}_n)+a_l\Tr(E_nA^{(l)}_n)+\sum_{\substack{k=1\\k\neq l}}^{2m}
\frac{a_k}{a_l-a_k}\Tr(A^{(k)}_n A^{(l)}_n),\quad l=1,\ldots,2m. \label{Lax6}
\eeq
\end{prop}

Now we can consider the linear system of ordinary differential equations
for the function $\Psi_n(z;a_1\ldots,a_{2m})$
\beq
\frac{\partial\Psi_n}{\partial z}=
\left\{E_n+\frac{1}{z}\left(C_n-\sum_{j=1}^{2m}A_n^{(j)}\right)+
\sum_{j=1}^{2m}\frac{A_n^{(j)}}{z-a_j}\right\}\Psi_n\label{Lax31}
\eeq
and
for $1\leq j\leq 2m$
\beq
\frac{\partial\Psi_n}{\partial a_j}=-\frac{A_n^{(j)}}{z-a_j}\Psi_n, \label{Lax41}
\eeq
where the poles $a_1,\ldots,a_{2m}$ play the role of deformation parameters.

The compatibility conditions for the system (\ref{Lax31}-\ref{Lax41}) lead
to Schlesinger equations which exactly coincide with equations of motion
(\ref{Lax3}-\ref{Lax5}). Therefore, we derive the isomonodromic deformation
representation for finite $n$ similar to
Jimbo, Miwa, Mori, Sato theory \cite{JMMS80}.

Now let us consider the case $m=1$ in more details.
We choose  $J=(0,s)$, $a_1=0$, $a_2=s$ and
introduce variables
\beq
x_i(s)=x_{i,2}^{(n)}(s), \quad y_i(s)=y_{i,2}^{(n)}(s),\quad \xi_i(s)=\xi_i^{(n)}(s),\quad
\eta_i(s)=\eta_i^{(n)}(s). \label{Lax42}
\eeq
We have only one nontrivial Hamiltonian $H_n(s)$ \eqref{ker134} and
the equations (\ref{Lax3}-\ref{Lax5}) for $l=2$
can be rewritten as
\begin{align}
s\frac{\p}{\p s}A^{(2)}_n=[C_n+s E_n,A^{(2)}_n],\quad
\frac{\p}{\p s}C_n=[E_n,A^{(2)}_n].\label{Lax43}
\end{align}
The isomonodromic equations (\ref{Lax31}-\ref{Lax41}) for $\Psi_n(z,s)$
rewrite as
\beq
\frac{\partial\Psi_n}{\partial z}=
\left\{E_n+\frac{C_n-A_n^{(2)}}{z}+
\frac{A_n^{(2)}}{z-s}\right\}\Psi_n,\quad
\frac{\partial\Psi_n}{\partial s}=-\frac{A_n^{(2)}}{z-s}\Psi_n.\label{Lax44}
\eeq

Using representation \eqref{Lax43} it is not difficult to construct
additional $M+1$ conserved quantities.

\begin{prop}\label{prop7}
The eigenvalues ${\it Spec}(A_n^{(2)}-C_n)$ are integrals of the motion
 \eqref{Lax43}
\end{prop}
\begin{proof}

Let us denote $B_n(s)=A_n^{(2)}(s)-C_n(s)$ and introduce the characteristic polynomial
\beq
p(z,s)=\det\left(z I-B(s)\right). \label{Lax46}
\eeq
We have
\begin{align}
\frac{\p}{\p s}p(z,s)=p(z,s)\frac{\p}{\p s}\Tr\log\left(z I-B_n(s)\right)
=-p(z,s)\sum_{k=0}^\infty\frac{1}{z^{k+1}}
\Tr\left(B^k_n(s)B'_n(s)\right),
\label{Lax47}
\end{align}
where the sum always converges for sufficiently large $z$.

Now
\beq
B'_n(s)=(A_n^{(2)}(s)-C_n(s))'=\frac{1}{s}[C_n,A_n^{(2)}]=
\frac{1}{s}[C_n,B_n]\label{Lax48}
\eeq
by \eqref{Lax43} and we have for any $k\geq0$
\beq
\Tr\left(B^k_n(s)B'_n(s)\right)=\frac{1}{s}\Tr\left(B^k_n[C_n,B_n]\right)=0. \label{Lax49}
\eeq
Therefore, the sum in \eqref{Lax47} is equal to zero and the characteristic polynomial
$p(z,s)$ does not depend on $s$.
We conclude that the eigenvalues of the matrix $B_n(s)$ do not depend on $s$.

\end{proof}
One can calculate the spectrum of the matrix $A_n^{(2)}-C_n$ using the small $s$
expansion.
Let us  assume that
\beq
\nu_{\mbox{\footnotesize min}}=\min(\nu_1,\ldots,\nu_M)>0. \label{Lax49d}
\eeq
We can use the integral representation
\eqref{ker14} for $Q_n(x)$ and calculate the integral by closing the contour to the left.
Using \eqref{ker114} we obtain
\beq
y_i(s)\sim s^{\nu_{\mbox{\footnotesize min}}}(\mbox{const}+\mbox{possible $\log$ terms})
\rightarrow 0 \quad \mbox{at}\quad
s\rightarrow 0. \label{Lax49a}
\eeq
Similarly, using the representation for $P_n(x)$ in terms of hypergeometric function
\eqref{A6} we obtain
\beq
x_0(0)\sim\mbox{const},\quad x_i(0)=0, \quad i>0.\label{Lax49b}
\eeq
Finally, $V_{i,j}^{(n)}(s)\rightarrow 0$ in \eqref{ker73} at $s\rightarrow 0$ and we obtain
from \eqref{ker117}
\beq
\xi_i(0)=(-1)^{M+i+1}
e_{M+1-i}(\boldsymbol{\nu}),\quad \eta_i(0)=0.\label{Lax49c}
\eeq

Therefore, the matrix $B_n(s)$ at $s=0$ becomes
\beq
B_n(0)=-
C_n(0)=\left(\begin{array}{lllll}
0 & 1 & 0 & \ldots & 0 \\
\vdots & \vdots & \ddots  & \, & \vdots \\
0 & 0 & 0 & \ldots & 1\\
0& d_1 & d_2 & \ldots & d_M
\end{array}\right),
\label{Lax50}
\eeq
where $d_i=(-1)^{M+i}e_{M+1-i}(\boldsymbol{\nu})$.

The matrix \eqref{Lax50} has the eigenvalue $0$ with the eigenvector $(1,0,\ldots,0)$
and $M$ eigenvalues $\nu_i$ with the eigenvectors $(1,\nu_i,\nu_i^2,\ldots,\nu_i^M)$
which follows from the identity
\beq
\sum_{k=1}^M (-1)^{M+k}e_{M+1-k}(\nu)\nu_i^k=(-\nu_i)(-\nu_i^M+\prod_{j=1}^M(\nu_i-\nu_j))=\nu_i^{M+1}.
\label{Lax51}
\eeq
 Therefore, we conclude that
 \beq
 {\rm Spec}(A_n^{(2)}-C_n)=(0,\nu_1,\ldots,\nu_M).\label{Lax52}
 \eeq
If $\nu_{\mbox{\footnotesize min}}=0$ in \eqref{Lax49d}, then $y_i(s)$ can have constant or growing logarithmic
asymptotics at $s\rightarrow0$ and the calculation becomes more involved.

A topic for future study is a possible relationship of this isomonodromic deformation with the theory of so-called
four dimensional Painlev\'e systems \cite{KNS12}, as speculated in \cite{WF17}.

\nsection{The case $M=1$}

We set $m=1$, choose $J=(0,s)$ and use the variables \eqref{Lax42}.
The system of partial differential equations (\ref{ker118}-\ref{ker124}) for $M=1$
reads
\begin{align}
sx_0'&=-\left(nx_0+x_1\right)\eta_0-x_1,\label{M2}\\
sx_1'&=-\left(nx_0+x_1\right)(\eta_1-s)+\xi_0x_0+\xi_1x_1,\label{M3}\\
sy_1'&=-\left(s+\xi_1\right)y_1+y_0+\eta_0y_0+\eta_1y_1,\label{M4}\\
sy'_0&=-\left(ns+\xi_0\right)y_1+n(\eta_0y_0+\eta_1y_1),\label{M5}\\
\xi_0'&=\phantom{-}(nx_0+x_1)y_0,\label{M6}\\
\xi_1'&=\phantom{-}(nx_0+x_1)y_1,\label{M7}\\
\eta_0'&=\phantom{-}x_0y_1,\label{M8}\\
\eta_1'&=\phantom{-}x_1y_1,\label{M9}
\end{align}
and the Hamiltonian \eqref{Lax6} for $l=2$ takes the form
\beq
H_2^{(n)}=(nx_0+x_1)(sy_1-\eta_0y_0-\eta_1y_1)-x_1y_0+y_1(\xi_0x_0+\xi_1x_1). \label{M10}
\eeq

Let us compare this with the results of  \cite{TW94c} for the Laguerre kernel.
From the Christoffel--Darboux formula (see e.g.~\cite[Prop.~5.1.3]{Forrester2010})
we define the Tracy-Widom kernel for the finite Laguerre ensemble of $n\times n$ matrices
by
\beq
K_{L}(x,y)=\frac{\varphi(x)\psi(y)-\psi(x)\varphi(y)}{x-y},\label{M12}
\eeq
with
\beq
\varphi(x)=\sqrt{\lambda}(n(n+\nu))^{1/4}\varphi_{n-1}(x),\quad
\psi(x)=\sqrt{\lambda}(n(n+\nu))^{1/4}\varphi_{n}(x),\label{M13}
\eeq
\beq
\varphi_n(x)=\sqrt{\frac{n!}{\Gamma(n+\nu+1)}}x^{\nu/2}e^{-\nu/2}L_n^{\nu}(x),\label{M14}
\eeq
as in the eqs. (1.2), (5.36) of \cite{TW94c}. $L_n^{\nu}(x)$ are generalized Laguerre
polynomials.

Taking into account the remark before \eqref{renorm}
the functions $P_n(x)$ and  $Q_n(x)$ in (\ref{ker3}-\ref{ker2}) for $M=1$ take the form
\beq
P_n(x)=(-1)^n n! L_n^\nu(x),\quad Q_n(x)=\la\frac{(-1)^n x^\nu e^{-x}}{\Gamma(n+\nu+1)}L_n^\nu(x),
\label{M15}
\eeq
where we used the expression of Laguerre polynomials in terms of the hypergeometric function
and (\ref{A5}) of the Appendix
\beq
L_n^\nu(x)=\frac{(\nu+1)_n}{n!}\>\ff
\left.\left(\begin{array}{c}-n \\1+\nu\end{array}\right|x\right).\label{M16}
\eeq
Now
we can rewrite
the formula (\ref{repr}) for the kernel $K_n^M(x,y)$  at $M=1$ as
\beq
K_n^1(x,y)=\frac{1}{x-y}\Big ( \delta_y-\delta_x+y-\nu\Big ) P_n(x)Q_n(y). \label{M17}
\eeq
Using the differentiation formula for the Laguerre polynomials
\beq
x\frac{d}{dx}L_n^\nu(x)=n L_n^\nu(x)-(n+\nu)L_{n-1}^\nu(x) \label{M18a}
\eeq
and substituting \eqref{M15} into \eqref{M17} we find after straightforward calculations
\beq
K_{L}(x,y) = \> h(x) K_n^1(x,y)  h^{-1}(y), \quad h(x)=x^{\nu/2}e^{-x/2}.\label{M18}
\eeq
We thus see that the kernel $K_n^1(x,y)$ is
symmetrizable, and after the diagonal similarity transformation coincides with
the symmetric kernel $K_{L}(x,y)$.

The functions $P_{L}(x)$  and $Q_{L}(x)$ of \cite[Eq.~(1.5)]{TW94c} match with our $Q_i^{(n)}(x;s)$ in
(\ref{ker71}) according to
\beq
P_{L}(x)= (1-K_{L})^{-1}\psi(x)=c_n \sqrt{\la}x^{\nu/2}e^{-x/2}Q_0^{(n)}(x;s), \label{M19}
\eeq
\beq
Q_{L}(x)=(1-K_{L})^{-1}\varphi(x)=\frac{c_n}{\sqrt{n(n+\nu)}}\sqrt{\la}x^{\nu/2}e^{-x/2}
(nQ_0^{(n)}(x;s)+Q_1^{(n)}(x;s)),\label{M20}
\eeq
\beq
c_n=\frac{(-1)^{n+1}(n(n+\nu))^{1/4}}{\sqrt{n!\Gamma(n+\nu+1)}}. \label{M21}
\eeq
Setting $x=s$ we find from (\ref{M19}-\ref{M20})
a connection of Tracy and Widom's variables $q(s)$ and $p(s)$ (1.6), \cite{TW94c}
with our variables $x_0(s)$ and $x_1(s)$
\beq
q(s)=-i\frac{c_n}{\sqrt{n}(n+\nu)}\sqrt{\la}s^{\nu/2}e^{-s/2}
(nx_0(s)+x_1(s)),\quad
p(s)=-ic_n \sqrt{\la}s^{\nu/2}e^{-s/2}x_0(s).\label{M22}
\eeq

We also obtain from \eqref{M15} and \eqref{M18} that
\beq
Q_n(x)=\frac{\la h^2(x)}{n!\Gamma(n+\nu+1)}P_n(x)\label{M23}
\eeq
and
\beq
K_n^1(y,x) = h^2(x) K_n^1(x,y)  h^{-2}(y). \label{M23a}
\eeq
It immediately follows that for $M=1$ we can express variables $y_i(s)$ in terms of $x_i(s)$, $i=0,1$
\beq
y_1(s)=\frac{\la s^{\nu}e^{-s}}{n! \Gamma(n+\nu+1)}x_0(s),\quad
y_0(s)=-\frac{\la s^{\nu}e^{-s}}{n! \Gamma(n+\nu+1)}x_1(s). \label{M11}
\eeq
We can now calculate correct initial conditions  for the variables \eqref{Lax42} for small $s$.
Assuming that the parameter $\nu>0$ is in generic position
 we obtain from (\ref{ker6}, \ref{ker71}-\ref{ker73}, \ref{ker114}-\ref{ker117})
 and explicit expressions (\ref{M15}-\ref{M16}) for $P_n(x)$, $Q_n(x)$
\beq
x_0(s)=i(-1)^{n+1}(\nu+1)_n\Big(1+O(s)\Big)
-\frac{i\la(-1)^n (\nu+1)_n^2s^{\nu+1}}{(\nu+1)\Gamma(n)\Gamma(\nu+2)}\big(1+O(s)
\big)+O(s^{2\nu+2}), \label{M11a}
\eeq
\beq
x_1(s)=i(-1)^{n+1}n(\nu+2)_{n-1}\big(s+O(s^2)\big)
-\frac{i\la(-1)^n n(\nu+1)_{n}^2s^{\nu+2}}{(\nu+1)\Gamma(n)\Gamma(\nu+3)}\big(1+O(s)\big)+O(s^{2\nu+3}),\label{M11b}
\eeq
\begin{align}
&\eta_0(s)=-\frac{(\nu+1)_n}{\Gamma(n+1)\Gamma(\nu+2)}\la s^{\nu+1}
\big(1+O(s)
\big)+  O(s^{2 \nu + 2}),\label{M11c}\\
&\eta_1(s)=-\frac{(\nu+1)_n}{\Gamma(n)\Gamma(\nu+3)}\la s^{\nu+2}
\big(1+O(s)\big)+  O(s^{2 \nu + 3}) ,\label{M11d}\\
&\xi_0(s)=\frac{ n(\nu+1)_n}{\Gamma(n)\Gamma(\nu+3)}\la s^{\nu+2}
\big(1+O(s)\big)+  O(s^{2 \nu + 3}),\label{M11e}\\
&\xi_1(s)=-\nu-\frac{(\nu+1)_n}{\Gamma(n)\Gamma(\nu+2)}\la s^{\nu+1}
\big(1+O(s)\big)+  O(s^{2 \nu + 2}) .\label{M11f}
\end{align}
So the recipe to restore a dependence on $\la$ in initial conditions
for primary variables is very simple --- we replace each power
$s^\nu$ by $\la s^\nu$.

With given initial conditions we can combine (\ref{M7}-\ref{M9}) to obtain the first
integral
\beq
\xi_1(s)-n\eta_0(s)-\eta_1(s)+\nu=0.\label{M32a}
\eeq
The second integral was derived in \eqref{ker137}
\beq
 H_2^{(n)}(s)-n\eta_0(s)-\eta_1(s)=0. \label{M32}
\eeq

Now let us see how integrals \eqref{Lax52} are expressed in terms of
basic variables. We have
\beq
A_n^{(2)}-C_n=\left(\begin{array}
{cc} x_0y_0+n\eta_0 & x_0y_1+\eta_0+1\\
x_1y_0+n\eta_1-\xi_0& x_1y_1+\eta_1-\xi_1
\end{array}\right)\label{M32b}
\eeq
The first integral $\Tr(A_n^{(2)}-C_n)=\nu$ is equivalent to \eqref{M32a}
due to the orthogonality condition
\beq
x_0(s)y_0(s)+x_1(s)y_1(s)=0 \label{M32c}
\eeq.
The second integral $\det(A_n^{(2)}-C_n)=0$ gives
\beq
(x_0y_0+n\eta_0)(x_1y_1+\nu-n\eta_0)-(x_0y_1+\eta_0+1)(x_1y_0+n\eta_1-\xi_0)=0.\label{M32d}
\eeq
Taking the sum of \eqref{M32} and \eqref{M32d} and using (\ref{M8}-\ref{M9}, \ref{M11})
we obtain the expression for $\xi_0(s)$ in terms of $\eta_0(s)$ and $\eta_1(s)$
\beq
\xi_0(s)=\frac{s(n\eta_0'(s)+\eta_1'(s))+(n\eta_0(s)-1)(n\eta_0(s)+\eta_1(s))+
n(\eta_1(s)-\nu\eta_0(s))}{1+\eta_0(s)}.\label{M32e}
\eeq
Using (\ref{M6}-\ref{M9}) to eliminate $x_0,x_1,y_0,y_1$ we can rewrite
the integral \eqref{M32d} as
\beq
n\eta_0(n\eta_0+\eta_1-\nu)+n\eta_1(1+\eta_0')+(n+\nu-n\eta_0)\eta_1'
-\xi_0(1+\eta_0+\eta_0')+(1+\eta_0)\xi_0'=0.
\label{M32f}
\eeq
Finally, differentiating (\ref{M8}-\ref{M9}) and using (\ref{M2}-\ref{M9})
it is easy to obtain the relations
\beq
s\eta_0''+2(1+\eta_0)\eta_1'+(2n\eta_0+s-\nu)\eta_0'=0,\label{M32g}
\eeq
\beq
s\eta_1''-(1+\eta_0)(n\eta_1'+\xi_0')+(n\eta_1-ns-\xi_0)\eta_0'=0.\label{M32h}
\eeq
Let us introduce the function
\beq
\sigma(s)=-n\eta_0(s)-\eta_1(s). \label{M32i}
\eeq
Using \eqref{M32e} to exclude $\xi_0$, $\xi_0'$ from (\ref{M32f}-\ref{M32h}) we
can express $\sigma$, $\sigma'$ and $\sigma''$ in terms of $\eta_0$, $\eta_0'$, $\eta_0''$
and obtain the 3rd order differential equation for $\eta_0(s)$.
The function $\sigma(s)$ satisfies
the
 $\sigma$-version of Painlev\'e V
\beq
(s\sigma'')^2=4s(\sigma')^3-4\sigma(\sigma')^2+\sigma^2+2(\nu-s+2n)\sigma\sigma'+
((\nu-s)^2-4sn)(\sigma')^2\label{M33}
\eeq
subject to the boundary condition
\beq
\sigma(s)=\frac{(\nu+1)_n}{\Gamma(n)\Gamma(\nu+2)}
\la s^{\nu+1}\bigg(1-\frac{2n+\nu}{\nu+2}s+\frac{\nu(\nu+1)^2+2n(n+\nu)(2\nu+3)}{2(\nu+1)_3}s^2
-\ldots\bigg)+  O(s^{2 \nu + 2}) .\nonumber\\
\eeq

Now let us find a correspondence with Tracy-Widom variables
 $u(s)$, $w(s)$ and $R(s)$ given by (5.41-5.46) in \cite{TW94c}
\begin{align}
sq'(s)&=\phantom{-}\left(\frac{s-\nu}{2}-n\right)q(s)+\left(\sqrt{n(n+\nu)}+u(s)\right)p(s),\label{M24}\\
sp'(s)&=-\left(\frac{s-\nu}{2}-n\right)p(s)-\left(\sqrt{n(n+\nu)}-w(s)\right)q(s),\label{M25}\\
sR(s)&=(s-\nu-2n)q(s)p(s)+\left(\sqrt{n(n+\nu)}+u(s)\right)p^2(s)\nonumber\\
&+\left(\sqrt{n(n+\nu)}-w(s)\right)q^2(s),\label{M26}\\
(sR(s)&)'=q(s)p(s),\label{M27}
\end{align}
and
\beq
u'(s)=q^2(s),\quad w'(s)=p^2(s). \label{M28}
\eeq
After straightforward calculations we obtain that the equations (\ref{M2}-\ref{M10})
are consistent with (\ref{M24}-\ref{M28}) under the choice
\begin{align}
u(s)&=-\frac{n(\xi_1(s)+\nu)-\xi_0(s)}{\sqrt{n(n+\nu)}},\label{M29}\\
w(s)&=-\sqrt{n(n+\nu)}\eta_0(s),\label{M30}\\
 sR(s)&=\sigma(s)=-H_2(s)=-n\eta_0(s)-\eta_1(s). \label{M31}
\end{align}

\nsection{The case $M=2$}
Again we consider the case $J=(0,s)$ with basic variables
$x_j$, $y_j$, $\eta_j$ and $\xi_j$, $j=0,1,2$ defined by \eqref{Lax42}.
The system (\ref{ker118}-\ref{ker124}) is
\begin{align}
sx_0'&=-\left(nx_0+x_1\right)\eta_0-x_1,\label{MM10}\\
sx_1'&=-\left(nx_0+x_1\right)\eta_1-x_2,\label{MM11}\\
sx_2'&=-\left(nx_0+x_1\right)(\eta_2+s)+\xi_0x_0+\xi_1x_1+\xi_2x_2,\label{MM12}\\
sy'_0&=-\left(\xi_0-ns\right)y_2+n(\eta_0y_0+\eta_1y_1+\eta_2y_2),\label{MM13}\\
sy_1'&=-\left(\xi_1-s\right)y_2+y_0+\eta_0y_0+\eta_1y_1+\eta_2y_2,\label{MM14}\\
sy_2'&=-\phantom{(}\xi_2y_2+y_1,\label{MM15}\\
\xi_0'&={-}(nx_0+x_1)y_0,\label{MM16}\\
\xi_1'&={-}(nx_0+x_1)y_1,\label{MM17}\\
\xi_2'&={-}(nx_0+x_1)y_2,\label{MM18}\\
\eta_0'&={-}x_0y_2,\label{MM19}\\
\eta_1'&={-}x_1y_2,\label{MM20}\\
\eta_2'&={-}x_2y_2,\label{MM21}
\end{align}

For simplicity we assume that $\nu_{1}, \nu_2>0$ are generic, i.e.
$\nu_1,\nu_2,\nu_1-\nu_2\not\in\mathbb{Z}$. Although from the random matrix application
the case of integer $\nu_1$, $\nu_2$ case is exactly the case of interest,
these restrictions can be lifted in principle by use of a limiting procedure.

We start with the initial conditions for basic variables at $s=0$.
First, biorthogonal functions  (\ref{ker3}-\ref{ker2})
can be expressed in terms of  generalized hypergeometric functions
\begin{align}
P_n(x)=(-1)^n(\nu_1+1)_n(\nu_2+1)_n\> \fm{1}{2}
\left.\left(\begin{array}{c}-n \\1+\nu_1,1+\nu_2\end{array}\right|x\right),\label{exp1}
\end{align}
\begin{align}
Q_n(x)=\frac{(-1)^n\la x^{\nu_1}\Gamma(\nu_2-\nu_1)}{n!\Gamma(\nu_1+1)\Gamma(\nu_2+n+1)}
\fm{1}{2}
\left.\left(\begin{array}{c}\nu_1+n+1 \\1+\nu_1,1+\nu_1-\nu_2\end{array}\right|x\right)
+
(\nu_1\leftrightarrow\nu_2),\label{exp2}
\end{align}
where \eqref{exp1} follows from \eqref{A6} and
\eqref{exp2} is obtained by closing the contour in \eqref{ker14} to the left
and summing over two series of poles.
Similar to the previous section a dependence on $\lambda$ is recovered by
replacing $s^{\nu_1}\to\la s^{\nu_1}$ and $s^{\nu_2}\to\la s^{\nu_2}$.

After straightforward calculations we obtain
\begin{align}\label{MM22}
x_0(s)=&i(-1)^{n+1}\la s^{\nu1+1}\frac{(\nu_1+2)_{n-1}^2(\nu_2+1)_n\Gamma(\nu_2-\nu_1)}
{\Gamma(n)\Gamma(\nu_1+1)\Gamma(\nu_2+1)}(1+O(s))+(\nu_1\leftrightarrow\nu_2)\nonumber\\
&+i(-1)^{n+1}(\nu_1+1)_n(\nu_2+1)_n
(1+O(s))
+O_{\nu_1+1,\nu_2+1},
\end{align}
\begin{align}\label{MM23}
x_1(s)&=i(-1)^{n+1}\la s^{\nu1+2}\frac{n(\nu_1+1)_{n}^2(\nu_2+1)_n\Gamma(\nu_2-\nu_1)}
{(\nu_1+1)\Gamma(n)\Gamma(\nu_1+3)\Gamma(\nu_2+2)}(1+O(s))+(\nu_1\leftrightarrow\nu_2)\nonumber\\
&+i(-1)^{n+1}n(\nu_1+2)_{n-1}(\nu_2+2)_{n-1}
(s+O(s^2))
+O_{\nu_1+1,\nu_2+1},
\end{align}
\begin{align}\label{MM24}
x_2(s)=&i(-1)^n\la s^{\nu1+2}\frac{n(\nu_1+1)_{n}^2(\nu_2+1)_n\Gamma(\nu_2-\nu_1)}
{(\nu_1+1)\Gamma(n)\Gamma(\nu_1+3)\Gamma(\nu_2+2)}(1+O(s))+(\nu_1\leftrightarrow\nu_2)\nonumber\\
+&i(-1)^nn(\nu_1+2)_{n-1}(\nu_2+2)_{n-1}
(s+O(s^2))
+O_{\nu_1+1,\nu_2+1},
\end{align}
\begin{align}
&y_0(s)=i\la s^{\nu_1+1}\frac{(-1)^n\Gamma(\nu_2-\nu_1)}
{\Gamma(n)\Gamma(\nu_1+2)\Gamma(\nu_2+n+1)}(1+O(s))+(\nu_1\leftrightarrow\nu_2)
+O_{\nu_1+1,\nu_2+1},\label{MM25}\\
&y_1(s)=-i\la s^{\nu_1}\frac{(-1)^n\nu_2\Gamma(\nu_2-\nu_1)}
{\Gamma(n+1)\Gamma(\nu_1+1)\Gamma(\nu_2+n+1)}(1+O(s))+(\nu_1\leftrightarrow\nu_2)
+O_{\nu_1,\nu_2},\label{MM26}\\
&y_2(s)=\phantom{,,}i\la s^{\nu_1}\frac{(-1)^n\Gamma(\nu_2-\nu_1)}
{\Gamma(n+1)\Gamma(\nu_1+1)\Gamma(\nu_2+n+1)}(1+O(s))+(\nu_1\leftrightarrow\nu_2)
+O_{\nu_1,\nu_2},\label{MM27}
\end{align}
\begin{align}
&\eta_0(s)=-\la s^{\nu_1+1}\frac{(\nu_1+2)_{n-1}\Gamma(\nu_2-\nu_1)}
{\Gamma(n+1)\Gamma(\nu_1+1)\Gamma(\nu_2+1)}(1+O(s))+(\nu_1\leftrightarrow\nu_2)
+O_{\nu_1+1,\nu_2+1},\label{MM28}\\
&\eta_1(s)=-\la s^{\nu_1+1}\frac{(\nu_1+1)_n\Gamma(\nu_2-\nu_1)}
{\Gamma(n)\Gamma(\nu_1+3)\Gamma(\nu_2+2)}(s+O(s^2))+(\nu_1\leftrightarrow\nu_2)
+O_{\nu_1+1,\nu_2+1},\label{MM29}\\
&\eta_2(s)=\phantom{,,}\la s^{\nu_1+1}\frac{(\nu_1+1)_n\Gamma(\nu_2-\nu_1)}
{\Gamma(n)\Gamma(\nu_1+3)\Gamma(\nu_2+2)}(s+O(s^2))+(\nu_1\leftrightarrow\nu_2)
+O_{\nu_1+1,\nu_2+1}\label{MM30}
\end{align}
\begin{align}
&\xi_0(s)=-\la s^{\nu_1+1}\frac{(\nu_1+1)_n\,n\Gamma(\nu_2-\nu_1)}
{\Gamma(n)\Gamma(\nu_1+3)\Gamma(\nu_2+1)}(s+O(s^2))+(\nu_1\leftrightarrow\nu_2)
+O_{\nu_1+1,\nu_2+1},\label{MM31}\\
&\xi_1(s)=e_2+\la s^{\nu_1+1}\frac{(\nu_1+1)_n\Gamma(\nu_2-\nu_1)}
{\Gamma(n)\Gamma(\nu_1+2)\Gamma(\nu_2)}(1+O(s))+(\nu_1\leftrightarrow\nu_2)
+O_{\nu_1+1,\nu_2+1},\label{MM32}\\
&\xi_2(s)=-e_1{-}\la s^{\nu_1+1}\frac{(\nu_1+1)_n\Gamma(\nu_2-\nu_1)}
{\Gamma(n)\Gamma(\nu_1+2)\Gamma(\nu_2+1)}(1+O(s))+(\nu_1\leftrightarrow\nu_2)
+O_{\nu_1+1,\nu_2+1},\label{MM33}
\end{align}
where we used a notation for higher order terms
\beq
O_{\al,\beta}=O(s^{2\al},s^{\al+\be},s^{2\be}) \label{MM33a}
\eeq
and
as before $e_1=\nu_1+\nu_2$, $e_2=\nu_1\nu_2$.

For later convenience let us introduce new variables
\beq
\chi_{0}=n\,\eta_{0}+\eta_{1},\quad \chi_{1}=n\,\eta_{1}+\eta_{2}.\label{MM34}
\eeq
The Hamiltonian \eqref{ker134}
\beq
H_2^{(n)}=-(nx_0+x_1)(sy_2+\eta_0y_0+\eta_1y_1+\eta_2y_2)-x_1y_0-x_2y_1+
y_2(\xi_0x_0+\xi_1x_1+\xi_2x_2) \label{M10a}
\eeq
leads to the first integral
\beq
H_2^{(n)}-\chi_0=0.\label{MM35}
\eeq

The {\bf Proposition \ref{prop7}} together with
the equation \eqref{Lax52} gives three additional integrals. Similar
to the case $M=1$ we can combine them with \eqref{MM35}
and the orthogonality condition
\beq
x_0(s)y_0(s)+x_1(s)y_1(s)+x_2(s)y_2(s)=0 \label{MM36}
\eeq
to derive the expressions for $\xi_i$'s in terms of
$\eta_i$'s. To do that  we first express the variables
$y_0$, $y_1$, $x_0$, $x_1$, $x_2$ from the equations
(\ref{MM16}-\ref{MM21}) in terms of $\xi_i$'s, $\eta_i$'s and $y_2$
and substitute into \eqref{MM35} and \eqref{Lax52}.
The dependence on the variable $y_2$ drops out
and after some algebra we obtain
\begin{align}
\xi_2&=\chi_0-e_1,\label{MM37}\\
\xi_1&=e_2-(1+e_1)\chi_0+\chi_0^2+s\chi_0'+\chi_1,\label{MM38}\\
\xi_0&=\frac{1}{n(1+\eta_0)}\Big\{n\chi_0(\chi_0-1-\nu_1)(\chi_0-1-\nu_2)+
(\eta_2-\chi_1)\big(\xi_1+n(n+e_1-\chi_0)\big)\nonumber\\
&+
n\chi_1(\chi_0+n-2)-ns\chi_0'(1+e_1-3\chi_0)+ns(s\chi_0''+2\chi_1')\Big\},\label{MM39}
\end{align}
where we also used the variables \eqref{MM34} to simplify final expressions.

Now the integral \eqref{MM35} will give a complicated differential equation
for $\eta_0$, $\eta_1$ and $\eta_2$. Similar to the case $M=1$ we would like
to derive a closed differential equation for the $\tau$-function \eqref{ker138}
which is expressed in terms of $\chi_0(s)$. To do that we need another
integral of the system (\ref{MM10}-\ref{MM21}). We were able to find such an additional
integral and combining it with \eqref{MM35} and initial conditions (\ref{MM22}-\ref{MM33})
to derive after tedious calculations
a coupled system of differential equations for $\chi_0(s)$ and $\chi_1(s)$
\begin{align}
\chi_1'&\left[3\chi_1'+3s\chi_0''+2\chi_0'(3\chi_0-e_1)\right]+
\chi_0'(s^2\chi_0'''+(1-e_1)s\chi_0'')+\nonumber\\
+\chi_0\chi_0'&\left[3s\chi_0''+\chi_0'(3\chi_0-2e_1-1)-1\right]+
(e_2-s){\chi_0'}^2+3s{\chi_0'}^3=0,\label{MM40}
\end{align}
\begin{align}
&(n-1){\chi_0'}^2(\chi_0-s\chi_0')+{\chi_0'}^3\left[(1+e_1+e_2-s-
(2+e1)\chi_0+\chi_0^2)\chi_0+s^2{\chi_0''}\right]
\nonumber\\
&+{\chi_0'}^4\left[3s\chi_0-s(1+e_1)\right]
+2\chi_1(1-{\chi_0'}){\chi_0'}^2+{\chi_1'}^2(\chi_1'+3\chi_0\chi_0'-e_1\chi_0')+
s^2\chi_0'\chi_0''\chi_1''\nonumber\\
&+\chi_1'\left[{\chi_0'}^2(e_2-s+3\chi_0^2+3s\chi_0')-
\chi_0\chi_0'(1+(1+2e_1)\chi_0')-s^2{\chi_0''}^2\right]=0.\label{MM41}
\end{align}
Currently
we don't know how to derive this additional integral algebraically
in terms of isomonodromic formulation. The next natural step is to eliminate
the function $\chi_1$ from the system (\ref{MM40}-\ref{MM41}) and to obtain
the  differential equation for $\chi_0$ which coincides
with the logarithmic derivative of the gap probaiblity.
However, this 4th order differential equation
is enormous and we can not give it here. In the hard edge scaling limit its
simplified version was derived in \cite{WF17}.

Let us notice that for the case $M=1$ both functions $\eta_0(s)$ and $\eta_1(s)$
satisfy the third order differential equations but the equation for their
linear combination $\sigma(s)$  can be integrated once
and gives the second order equation \eqref{M33}. It is not clear whether the system
(\ref{MM40}-\ref{MM41}) can be integrated further to produce a simpler
third order differential equation for some combination of $\chi_0$ and $\chi_1$.
In the hard edge scaling limit with $\nu_1=-1/2$, $\nu_2=0$
such third order differential equation was found in \cite{WF17} ,
but it may be the case only at this special point.

With the given asymptotics at $s\to0$ which follow from (\ref{MM28}-\ref{MM30})
\begin{align}
&\chi_0(s)=\al_0 s^{\nu_1+1}(1+O(s))+\beta_0 s^{\nu_2+1}(1+O(s))+O_{\nu_1+1,\nu_2+1},\nonumber\\
&\chi_1(s)=\al_1 s^{\nu_1+1}(s+O(s^2))+\beta_1 s^{\nu_2+1}(s+O(s^2))+O_{\nu_1+1,\nu_2+1}\label{MM42},
\end{align}
the system (\ref{MM40}-\ref{MM41}) uniquely
determines power series expansions for $\chi_0$, $\chi_1$
in terms of two free parameters $\al_0$ and $\beta_0$. Parameters $\al_1$ and $\beta_1$ are fixed
by
\beq
\al_1=\frac{(n-1)\alpha_0}{(\nu_1+2)(\nu_2+1)},\quad \beta_1=\frac{(n-1)\beta_0}{(\nu_1+1)(\nu_2+2)}\label{MM43}
\eeq
and the power series for $\chi_0(s)$ and $\chi_1(s)$ have the form
\begin{align}
\chi_0(s)&=\al_0 s^{\nu_1+1}\left(1+\frac{2+2\nu_1+\nu_1\nu_2+n(2\nu_2-\nu_1)}
{(\nu_1+2)(\nu_2+1)(1+\nu_1-\nu_2)}s+O(s^2)\right)\nonumber\\
&+\beta_0 s^{\nu_2+1}\left(1+\frac{2+2\nu_2+\nu_1\nu_2+n(2\nu_1-\nu_2)}
{(\nu_1+1)(\nu_2+2)(1+\nu_2-\nu_1)}s+O(s^2)\right)\nonumber\\
&-\al_0\beta_0 s^{\nu_1+\nu_2+2}\frac{\nu_1+\nu_2+2}{(\nu_1+1)(\nu_2+1)}(1+O(s))+
O_{\nu_1+1,\nu_2+1}, \label{MM44}
\end{align}
\begin{align}
\chi_1(s)&= \alpha_0\frac{(n-1)s^{\nu_1+2}}{(\nu_1+2)(\nu_2+1)}
\left(1+O(s)\right)+\beta_0\frac{(n-1)s^{\nu_2+2}}{(\nu_1+1)(\nu_2+2)}
\left(1+O(s)\right)\nonumber\\
&-\al_0\beta_0s^{\nu_1+\nu_2+3}\frac{(n-1)(\nu_1+\nu_2+4)}{(\nu_1+1)(\nu_1+2)(\nu_2+1)(\nu_2+2)}
+O_{\nu_1+2,\nu_2+2}.\label{MM45}
\end{align}
The coefficients $\al_0$, $\beta_0$ are fixed by the asymptotics of $\eta_0(s)$ \eqref{MM28}
\beq
\al_0=-\frac{\la\,(\nu_1+2)_{n-1}\Gamma(\nu_2-\nu_1)}{\Gamma(n)\Gamma(\nu_1+1)\Gamma(\nu_2+1)},\quad
\beta_0=-\frac{\la\,(\nu_2+2)_{n-1}\Gamma(\nu_1-\nu_2)}{\Gamma(n)\Gamma(\nu_1+1)\Gamma(\nu_2+1)}.\label{MM46}
\eeq
The gap probability \eqref{ker133}
is given by
\beq
E_n^M(0;J)=\exp\left\{\int_{0}^s\chi_0(t)\frac{dt}{t}\right\} \label{MM47}
\eeq
and its expansion at $s=0$ has the form
\begin{align}
&E_n^M(0;J)=1+\al_0{s^{\nu_1+1}}\left(\frac{1}{\nu_1+1}+\frac{2+2\nu_1+\nu_1\nu_2+n(2\nu_2-\nu_1)}
{(\nu_1+2)^2(\nu_2+1)(1+\nu_1-\nu_2)}s+O(s^2)\right)\nonumber\\
&+\beta_0 s^{\nu_2+1}\left(\frac{1}{\nu_2+1}+\frac{2+2\nu_2+\nu_1\nu_2+n(2\nu_1-\nu_2)}
{(\nu_1+1)(\nu_2+2)^2(1+\nu_2-\nu_1)}s+O(s^2)\right)\nonumber\\
&-\al_0\beta_0 s^{\nu_1+\nu_2+3}\frac{(n-1)(\nu_1-\nu_2)^2}
{(\nu_1+1)^2(\nu_2+1)^2(\nu_1+2)^2(\nu_2+2)^2}(1+O(s))+
O_{\nu_1+1,\nu_2+1}. \label{MM48}
\end{align}
\section{
Acknowledgments}
We would like to thank V. Bazhanov and  J.R. Ipsen for useful discussions
and
N. Witte for careful reading of the manuscript and his comments.
We acknowledge support by the Australian Research Council through
grant DP140102613 (PJF, VVM) and the ARC Centre of Excellence for
Mathematical and Statistical Frontiers (PJF).

\appendixtitleon
\begin{appendices}
\numberwithin{equation}{section}

\renewcommand{\thesection}{} 
\section{}
\renewcommand{\thesection}{\Alph{section}}

In this appendix we give definitions for the generalized hypergeometric function
and for the Meijer G-function and discuss some of their properties. We follow
notations of \cite{Luke69}.

The generalized hypergeometric function is defined by a power series
\beq
\hypergeometric{p}{q}{a_1,\ldots,a_p}{b_1,\ldots,b_q}{z}=
{{\mathlarger{\ds\sum}_{n=0}^\infty}}\>
\frac{\ds\prod_{i=1}^p(a_i)_n}{\ds\prod_{j=1}^q(b_j)_n}\>\frac{z^n}{n!}.\label{A1}
\eeq
where the Pochhammer symbol is defined by
\beq
(a)_n=\frac{\Gamma(a+n)}{\Gamma(a)}=\prod_{k=0}^{n-1}(a+k), \quad n\geq0.\label{A2}
\eeq
We assume that $z$ is chosen in the region of convergence of \eqref{A1}. This region can
be extended by a contour integral representation like for the Meijer G-function below.

The Meijer G-function is given by a contour integral
\beq
\MeijerG{m}{n}{p}{q}{a_1,\ldots,a_p}{b_1,\ldots,b_q}{z}=
\frac{1}{2\pi i}\int_{L}
\frac{\ds \prod_{i=1}^m \Gamma(b_i-u)\prod_{i=1}^n\Gamma(1-a_i+u)}
{\ds\prod_{i=n+1}^p\Gamma(a_i-u)\prod_{i=m+1}^q\Gamma(1-b_i+u)}
\>z^u du, \label{A3}
\eeq
where $m,n,p,q$ are integers such that
$0\leq m \leq q$, $0\leq n\leq p$
and no pole of $\Gamma(b_j-u)$, $j=1,\ldots,m$ coincides with any pole of $\Gamma(1-a_k+u)$,
$k=1,\ldots,n$.

The contour $L$  runs
from $-i\infty$ to $+i\infty$ separating the poles of
$\Gamma(b_j-u)$, $j=1,\ldots,m$ on the right and $\Gamma(1-a_k+u)$,
$k=1,\ldots,n$ on the left. It can also be a loop starting and ending at $+\infty$ and encircling
poles of $\Gamma(b_j-u)$ for $p<q$ or a loop starting and ending at $-\infty$ and encircling
poles of $\Gamma(1-a_k+u)$ for $p>q$.

The Meijer G-function satisfies the differential equation
\beq
\left[(-1)^{p-m-n}\>z\prod_{j=1}^p(\delta_z-a_j+1)-\prod_{j=1}^q(\delta_z-b_j)\right]
\MeijerG{m}{n}{p}{q}{a_1,\ldots,a_p}{b_1,\ldots,b_q}{z}=0. \label{A4}
\eeq
Let us set $m=1$, $n=0$, $b_1=0$ and assume that $b_j\notin\mathbb{Z}$, $j=2,\ldots,q+1$  and $p\leq q+1$.
Then we can evaluate the integral in \eqref{A3} over the loop starting and ending at $\infty$
and encircling poles $\Gamma(-u)$. The sum of the residues gives the generalized hypergeometric function
and we get the relation
\begin{align}
\hypergeometric{p}{q}{a_1,\ldots,a_p}{b_1,\ldots,b_q}{z}=
\prod_{i=1}^p\Gamma(1-a_i)\prod_{j=1}^q\Gamma(b_j)\>
\MeijerG{1}{0}{p}{q+1}{1-a_1,\ldots,1-a_p}{0,1-b_1,\ldots,1-b_q}{(-1)^{p+1}z}.\label{A5}
\end{align}
If any of $a_i$, $i=1,\ldots,p$ is equal to a negative integer $-n$, the hypergeometric
series truncates and we get a polynomial of the degree $n$. In particular,
setting $p=1$, $q=M$ in \eqref{A5} and comparing with \eqref{ker2} we obtain
a representation of polynomials $P_n(x)$ in terms of the generalized hypergeometric function $
\fm{1}{M}$
\begin{align}
P_n(x)=(-1)^n\prod_{j=1}^M(\nu_j+1)_n\>\> \fm{1}{M}
\left.\left(\begin{array}{c}-n \\1+\nu_1,\ldots,1+\nu_M\end{array}\right|x\right).\label{A6}
\end{align}

\end{appendices}

\bibliographystyle{plain}

\begin{thebibliography}{10}

\bibitem{St14}
E.~Strahov, ``Differential equations for singular values of products of
  {G}inibre random matrices,'' {\em J. Phys. A} {\bfseries 47} no.~32, (2014)
  325203, 27.

\bibitem{Forrester2010}
P.~J. Forrester, {\em Log-gases and random matrices}, vol.~34 of {\em London
  Mathematical Society Monographs Series}.
\newblock Princeton University Press, Princeton, NJ, 2010.

\bibitem{Ka15}
M.~Katori, {\em Bessel processes, {S}chramm-{L}oewner evolution, and the
  {D}yson model}, vol.~11 of {\em SpringerBriefs in Mathematical Physics}.
\newblock Springer, [Singapore], 2015.

\bibitem{IIKS90}
A.~R. Its, A.~G. Izergin, V.~E. Korepin, and N.~A. Slavnov, ``Differential
  equations for quantum correlation functions,''
  \href{http://dx.doi.org/10.1142/S0217979290000504}{{\em Internat. J. Modern
  Phys. B} {\bfseries 4} no.~5, (1990) 1003--1037}.

\bibitem{TW94c}
C.~A. Tracy and H.~Widom, ``Fredholm determinants, differential equations and
  matrix models,'' {\em Comm. Math. Phys.} {\bfseries 163} no.~1, (1994)
  33--72.

\bibitem{AvM07}
M.~Adler and P.~van Moerbeke, ``PDEs for the Gaussian ensemble with external source and
the Pearcey distribution", Comm. Pure Appl. Math. {\bfseries 60}, (2007) 1261--1292

\bibitem{Bo09}
F.~Bornemann, \emph{On the numerical evaluation of {F}redholm determinants}, Math.
  Comp. \textbf{79}, (2010) 871--915.

\bibitem{BC13}
M.~Bertola and M.~Cafasso, \emph{The gap probabilities of the tacnode, Pearcey and
Airy processes, their mutual relationship and evaluation},
Random Matrices: Theory Appl. {\bfseries 02}, (2013) 1350003  [18 pages].


\bibitem{JMMS80}
M.~Jimbo, T.~Miwa, Y.~M\^ori, and M.~Sato, ``Density matrix of an impenetrable
  {B}ose gas and the fifth {P}ainlev\'e transcendent,''
  \href{http://dx.doi.org/10.1016/0167-2789(80)90006-8}{{\em Phys. D}
  {\bfseries 1} no.~1, (1980) 80--158}.

\bibitem{McCoy76}
T.~T. Wu, B.~M. McCoy, C.~A. Tracy, and E.~Barouch, ``Spin-spin correlation
  functions for the two-dimensional {I}sing model: exact theory in the scaling
  region,'' {\em Phys. Rev. B} {\bfseries 13} (1976) 316--374.

\bibitem{Zam94}
{\relax Al}.~B. Zamolodchikov, ``Painleve III and $2$d polymers,'' {\em Nuclear
  Phys. B} {\bfseries 432} no.~3, (1994) 427--456.

\bibitem{BM06}
V.~V. Bazhanov and V.~V. Mangazeev, ``The eight-vertex model and {P}ainlev\'e
  {VI},'' {\em J. Phys. A} {\bfseries 39} no.~39, (2006) 12235--12243.

\bibitem{BH98}
E.~Br\'ezin and S.~Hikami, ``Level spacing of random matrices in an external
  source,'' \href{http://dx.doi.org/10.1103/PhysRevE.58.7176}{{\em Phys. Rev. E
  (3)} {\bfseries 58} no.~6, part A, (1998) 7176--7185}.

\bibitem{TW06}
C.~A. Tracy and H.~Widom, ``The {P}earcey process,''
  \href{http://dx.doi.org/10.1007/s00220-005-1506-3}{{\em Comm. Math. Phys.}
  {\bfseries 263} no.~2, (2006) 381--400}.

\bibitem{BK06}
P.~M. Bleher and A.~B.~J. Kuijlaars, ``Large {$n$} limit of {G}aussian random
  matrices with external source. {III}. {D}ouble scaling limit,''
  \href{http://dx.doi.org/10.1007/s00220-006-0159-1}{{\em Comm. Math. Phys.}
  {\bfseries 270} no.~2, (2007) 481--517}.

\bibitem{KZ14}
A.~B.~J. Kuijlaars and L.~Zhang, ``Singular values of products of {G}inibre
  random matrices, multiple orthogonal polynomials and hard edge scaling
  limits,'' \href{http://dx.doi.org/10.1007/s00220-014-2064-3}{{\em Comm. Math.
  Phys.} {\bfseries 332} no.~2, (2014) 759--781}.



\bibitem{WF17}
N.~Witte and P.~Forrester, ``Singular values of products of Ginibre random
  matrices,'' \href{http://dx.doi.org/10.1111/sapm.12147}{{\em Studies in
  Applied Mathematics} {\bfseries 138} no.~2, (2017) 135--184}.

\bibitem{CGS16}
T.~Claeys, M.~Girotti and D.~Stivigny,
``Large gap asymptotics at the hard edge for product random matrices and Muttalib-Borodin ensembles",
arXiv:1612.01916

\bibitem{TW94b}
C.~A. Tracy and H.~Widom, ``Level spacing distributions and the {B}essel
  kernel,'' {\em Comm. Math. Phys.} {\bfseries 161} no.~2, (1994) 289--309.

\bibitem{Fo93c}
P.~J. Forrester, ``The spectrum edge of random matrix ensembles,''
  \href{http://dx.doi.org/10.1016/0550-3213(93)90126-A}{{\em Nuclear Phys. B}
  {\bfseries 402} no.~3, (1993) 709--728}.

  \bibitem{AvM95}
M.~Adler and P.~van Moerbeke, ``Matrix integrals, {Toda} symmetries,
  {Virasora} constraints and orthogonal polynomials", Duke Math. Journal
  \textbf{80}, (1995)  863--911.

  \bibitem{FW01a}
P.~J.~Forrester and N.S.~Witte, ``Application of the $\tau$-function theory of {Painlev\'e}
  equations to random matrices: {PV}, {PIII}, the {LUE}, {JUE} and {CUE}",
  Commun. Pure Appl. Math. \textbf{55}, (2002) 679--727.

 \bibitem{FO10}
  P.J. Forrester and C.M. Ormerod, ``Differential equations for deformed {L}aguerre polynomials",
J. Approx. Th. {\bf 162}, (2010) 653--677.

\bibitem{BC09}
E.~Basor and Y.~Chen, ``{Painlev\'e} V and the distribution of a discontinuous linear statistic in the
{L}aguerre unitary ensembles" {\em J. Phys. A} {\bf 42}, (2009) 035203.

\bibitem{AIK13}
G.~{Akemann}, J.~R. {Ipsen}, and M.~{Kieburg}, ``{Products of rectangular
  random matrices: Singular values and progressive scattering},'' {\em
  Phys.Rev. E} {\bfseries 88} no.~5, (2013) 052118.

\bibitem{AKW13}
G.~Akemann, M.~Kieburg, and L.~Wei, ``Singular value correlation functions for
  products of {W}ishart random matrices,''
  \href{http://dx.doi.org/10.1088/1751-8113/46/27/275205}{{\em J. Phys. A}
  {\bfseries 46} no.~27, (2013) 275205, 22}.

   \bibitem{AI15}  G.~Akemann and J.R. Ipsen, ``Recent exact and asymptotic results for products
of independent random matrices" Acta Physica Polonica B {\bf 46},  (2015) 1747--1784.

\bibitem{BS13}
R.~Beals and J.~Szmigielski, ``Meijer {$G$}-functions: a gentle introduction,''
  \href{http://dx.doi.org/10.1090/noti1016}{{\em Notices Amer. Math. Soc.}
  {\bfseries 60} no.~7, (2013) 866--872}.

\bibitem{Luke69}
Y.~L. Luke, {\em The special functions and their approximations, {V}ol. {I}}.
\newblock Mathematics in Science and Engineering, Vol. 53. Academic Press, New
  York-London, 1969.

\bibitem{BGS14}
M.~Bertola,   M.~Gekhtman and  J.~Szmigielski, ``Cauchy-Laguerre two-matrix
model and the Meijer-G random point field", Commun. Math. Phys. \textbf{326}, (2014) 111--144.

\bibitem{FK16}
P.J.~Forrester and M.~Kieburg, ``Relating the Bures measure to the Cauchy two-matrix model",
Commun. Math. Phys. \textbf{342}, (2016) 151--187.

\bibitem{Fo14}
P.J.~Forrester, ``Eigenvalue statistics for product complex Wishart matrices", 
J. Phys. A \textbf{47}, (2014) 345202.

\bibitem{KKS16}
M.~Kieburg, A.~B.~J.~Kuijlaars and D.~Stivigny, ``Singular value statistics of 
matrix products with truncated unitary matrices",
Int. Math. Research Notices \textbf{2016}, (2016)  3392-3424.


\bibitem{Bor99}
A.~Borodin, ``Biorthogonal ensembles,'' {\em Nuclear Phys. B} {\bfseries 536}
  no.~3, (1999) 704--732.

\bibitem{Zh17} L.~Zhang, ``On Wright's generalized Bessel kernel", Physica D \textbf{340}, (2017) 27--39.

\bibitem{Mu95}
K.~A.~Muttalib, ``Random matrix models with additional interactions", J.
  Phys. A \textbf{28}, (1995) L159--L164.

\bibitem{FW15}
P.~J.~Forrester and D.~Wang, ``Muttalib--Borodin ensembles in random matrix
  theory --- realisations and correlation functions", arXiv:1502.07147.


  \bibitem{KNS12}  H. Kawakami, A. Nakamura, and H. Sakai.
  ``Toward a classification of four-dimensional Painlevé-type equations". In
A Dzhamay, K Maruno, and VU Pierce, editors, Algebraic and geometric aspects
of integrable systems and random matrices, volume 593 of Contemporary Mathematics,
pages 143--161. Amer. Math. Soc., Providence, RI, 2013.


\end{thebibliography}

\providecommand{\bysame}{\leavevmode\hbox to3em{\hrulefill}\thinspace}
\providecommand{\MR}{\relax\ifhmode\unskip\space\fi MR }
\providecommand{\MRhref}[2]{%
  \href{http://www.ams.org/mathscinet-getitem?mr=#1}{#2}
}
\providecommand{\href}[2]{#2}

\end{document}